\documentclass[letterpaper]{article}

\usepackage[submission]{aaai23}

\usepackage{times}
\usepackage{helvet}
\usepackage{courier}
\usepackage[hyphens]{url}
\usepackage{graphicx}
\usepackage{natbib}
\usepackage{caption}
\usepackage{algorithm}

\title{Regret Pruning for Learning Equilibria in Simulation-Based Games}
\author{
Bhaskar Mishra,
Cyrus Cousins,
Amy Greenwald
}
\affiliations{
    University of Florida,
    UMass Amherst,
    Brown University
}

\usepackage{booktabs}
\usepackage{todonotes}

\usepackage{bm}
\usepackage{color}
\usepackage{url}
\usepackage{graphicx}
\usepackage{pgfplots}
\usepackage{tikz}
\usepackage{algorithm}
\usepackage{algpseudocode}
\usepackage{mathrsfs}
\usepackage{relsize}
\usepackage{yfonts}
\usepackage{nicefrac}
\usepackage{multicol}
\usepackage{subcaption}
\usepackage{wrapfig}
\usepackage{sgame}

\usepackage{mathrsfs}
\usepackage{verbatim}
\usepackage{microtype}
\usepackage[normalem]{ulem}

\usepackage{amsmath,amssymb,amsfonts,amsthm}
\usepackage{mathtools}
\usepackage{thmtools}
\usepackage{thm-restate}
\usepackage{bbm}
\usepackage{dsfont}
\usepackage{hyperref}
\usepackage{cleveref}
\usepackage{balance}

\usepackage{upgreek}

\usepackage{footnote}

\usetikzlibrary{arrows, automata}

\pgfplotsset{compat=1.15}
\newcommand{\mydef}[1]{\emph{#1}}

\newfloat{algorithm}{t}{lop}
\floatname{algorithm}{Algorithm}
\algnewcommand{\Input}{\textbf{input:}~}
\algnewcommand{\Output}{\textbf{output:}~}
\algnewcommand{\Error}{\textbf{error}~}
\algnewcommand{\Continue}{\textbf{continue}~}
\algnewcommand{\Break}{\textbf{break}~}

\newcommand{\Samples}{\bm{Y}}
\newcommand{\SamplePoint}{\ConditionValue}

\newcommand{\NumberOfSamples}{m}

\newcommand{\UtilityRange}{c}
\newcommand{\UtilityVariance}{\bm{v}}
\newcommand{\EUtilityVariance}{\hat{\bm{v}}}
\newcommand{\ConditionSpace}{\mathcal{Y}}
\newcommand{\ConditionValue}{y}

\newcommand{\ConditionDistribution}{\rand{D}}

\newcommand{\TimeIndex}{t}

\newcommand{\ScheduleLength}{T}

\newcommand{\GameTuple}{\Gamma}
\newcommand{\ConditionalGame}[1]{\GameTuple_{#1}}    % Typical use: \ConditionalGame{\ConditionSpace}
\newcommand{\EmpiricalGame}[1]{\hat{\GameTuple}_{#1}} % Typical use: \EmpiricalGame{\Samples}

\newcommand{\NumberOfPlayers}{|\SetOfPlayers|}
\newcommand{\SetOfPlayers}{P}
\newcommand{\PlayerIndex}{p}
\newcommand{\StrategySet}{S}

\newcommand{\StrategyAlt}{t}
\newcommand{\StratProfile}{\bm{s}}

\newcommand{\StratProfileSpace}{\bm{\StrategySet}}

\newcommand{\MixedStrategySet}{{\StrategySet}^{\diamond}}

\newcommand{\MixedStratProfileSpace}{\bm{\MixedStrategySet}}

\newcommand{\EUtility}{\hat{\bm{u}}}
\newcommand{\Utility}{\bm{u}}

\DeclareMathOperator{\Adjacent}{Adj}

\DeclareMathOperator{\Regret}{Reg}

\newcommand{\UtilityIndices}{\bm{\mathcal{I}}}

\renewcommand{\epsilon}{\varepsilon}

\renewcommand{\delta}{\updelta}

\newcommand{\ssum}{\mathsmaller{\sum}\limits}

\newcommand{\rand}[1]{\mathscr{#1}}

\newcommand{\abs}[1]{\left\lvert{} #1 \right\rvert{}}
\newcommand{\norm}[1]{\left\lVert{} #1 \right\rVert{}}

\newcommand{\R}{\mathbb{R}}

\newif\iflsymb
\lsymbfalse

\iflsymb

\else

\fi

\allowdisplaybreaks

\DeclareMathOperator*{\Prob}{\mathbb{P}}

\DeclareMathOperator*{\Expect}{\mathbb{E}}

\DeclareMathOperator*{\Var}{\mathbb{V}}

\newcommand{\Simulator}{\mathscr{S}}
\newcommand{\SimUtility}{\Utility}
\newcommand{\SampleUtility}{\dot{\Utility}}
\newcommand{\NumberQueries}{m}
\newcommand{\epsilonH}{\epsilon^{\textup{H}}}
\newcommand{\epsilonEB}{\epsilon^{\hat{\textup{B}}}}
\newcommand{\EVarianceEpsilon}{\epsilon^{\EUtilityVariance}}
\newcommand{\Equilibria}{\textup{E}}
\newcommand{\MixedEquilibria}{\Equilibria^{\diamond}}

\newcommand{\PSWE}{PS-WE}
\newcommand{\PSRegOld}{PS-REG-0}
\newcommand{\PSRegU}{PS-REG}
\newcommand{\PSRegNU}{PS-REG+}
\newcommand{\PSRegM}{PS-REG-M}

\newcommand{\CumulativeSamples}{M}

\newtheorem{definition}{Definition}
\newtheorem{assumption}{Assumption}
\newtheorem{theorem}{Theorem}

\newtheorem{lemma}{Lemma}
\newtheorem{counterexample}{Counterexample}
\newtheorem{claim}{Claim}

\begin{document}

\maketitle

\begin{abstract}
In recent years, empirical game-theoretic analysis (EGTA) has emerged as a powerful tool for analyzing games in which an exact specification of the utilities is unavailable.
Instead, EGTA assumes access to an oracle, i.e., a 
%(stochastic) 
simulator, which can generate unbiased noisy samples of players' unknown utilities, given a strategy profile. 
Utilities can thus be empirically estimated by repeatedly querying the simulator.
Recently, various progressive sampling (PS) algorithms have been proposed,
%which aim to use the simulator to progressively sample utilities so that, 
which aim to produce PAC-style learning guarantees (e.g., approximate Nash equilibria with high probability) using as few simulator queries as possible.
% which aim to produce an empirical game whose approximate Nash equilibria are also good approximate Nash equilibria w.r.t. the true utilities, using as few simulator queries as possible.
A recent work by \citeauthor{areyan2020improved} introduces a pruning technique called \emph{regret-pruning} which further minimizes the number of simulator queries placed in PS algorithms which aim to learn \emph{pure} Nash equilibria.
In this paper, we address a serious limitation of this original regret pruning approach -- it is only able to guarantee that \emph{true} pure Nash equilibria of the empirical game are approximate equilibria of the true game, and is unable to provide any strong guarantees regarding the efficacy of approximate pure Nash equilibria. This is a significant limitation since in many games, pure Nash equilibria are computationally intractable to find, or even non-existent. We introduce three novel regret pruning variations. The first two variations generalize the original regret pruning approach to yield guarantees for approximate pure Nash equilibria of the empirical game. The third variation goes further to even yield strong guarantees for all approximate mixed Nash equilibria of the empirical game. We use these regret pruning variations to design two novel progressive sampling algorithms, \PSRegNU{} and \PSRegM{}, which experimentally outperform the previous state-of-the-art algorithms for learning pure and mixed equilibria, respectively, of simulation-based games.
\end{abstract}
\section{Introduction}

Game theory is the standard conceptual framework used to analyze strategic interactions among rational agents in multi-agent systems.
A game comprises a collection of players, each with a set of 
%\samy{available}{}
strategies and a utility function, mapping strategy profiles (i.e., combinations of strategies)
%(i.e., strategy profiles) 
to values.
Traditionally, game-theoretic analysis presumes complete access to a game's structure, including the utility functions.
%\sbhaskar{Even in games where utilities are inherently stochastic, it is assumed that expected utilities are known, or can be computed.}{} 

In recent years, empirical game-theoretic analysis (EGTA) has emerged as a powerful tool for analyzing games in which such an exact specification of the utilities is unavailable.
Instead, EGTA assumes access to an oracle, i.e., a (stochastic) simulator, which produces unbiased noisy samples of players' unknown 
%\emph{true} 
utilities given a strategy profile~\cite{Wellman06,tuyls2020bounds,areyan2020improved}.
Such games are called \mydef{simulation-based games}~\cite{vorobeychik2008stochastic}, or \mydef{black-box} games~\cite{picheny2016bayesian}, and their empirical counterparts, which are derived from simulation data, are called \mydef{empirical games}.
Simulation-based games have been studied in many practical settings including trading agent analyses in supply chains~\cite{vorobeychik2006empirical,jordan2007empirical}, ad auctions~\cite{jordan2010designing,DBLP:conf/uai/ViqueiraCMG19}, and
energy markets~\citep{ketter2013autonomous};
designing network routing protocols~\cite{wellman2013analyzing}; strategy selection in real-time games~\cite{tavares2016rock}; and 
the dynamics of 
%%% SPACE
%reinforcement learning 
RL algorithms, like AlphaGo \cite{tuyls2018generalised}. 

A typical EGTA goal is to produce PAC-style learning guarantees (e.g., approximate Nash equilibria with high probability) with minimal \emph{query complexity}, i.e., the number of simulation queries placed \cite{tuyls2020bounds,areyan2020improved,cousins2022computational}.
This goal has led to the development of \mydef{progressive sampling}
algorithms~\cite{areyan2020improved, cousins2022computational}, which place simulation queries in progressive batches, until the desired guarantee is reached.
%which aim to use the simulator to progressively sample utilities at various strategy profiles so as to, with high probability, produce an empirical game whose $\epsilon$-Nash equilibria are also good (commonly $2\epsilon$) approximate Nash equilibria in the corresponding simulation-based game. 
On top of progressive sampling, two papers introduce various pruning techniques, to further minimize query complexity.
One of these techniques, \mydef{well-estimated pruning}, prunes strategy profiles whose utilities are very likely to already be sufficiently close
%$\epsilon$-uniform approximations 
to the true utilities~\cite{cousins2022computational}.
This technique is useful for learning a variety of game properties, including regret, pure or mixed Nash equilibria, welfare-maximizing outcomes, and more.
A second technique, called \mydef{regret pruning}, is intended for use only when learning pure Nash equilibria, as it prunes strategy profiles that are highly unlikely to be best responses, and hence unlikely to be necessary for finding pure Nash equilibria ~\cite{areyan2020improved}.

In this paper, we focus predominantly on regret pruning. \citet{areyan2020improved} claim that their regret pruning criterion can be used to learn an empirical game which, with high probability, satisfies a certain dual pure Nash containment guarantee -- all pure Nash equilibria of the simulation-based game are approximate equilibria of the empirical game, and all approximate pure Nash equilibria of the empirical game are approximate equilibria of the simulation-based game. We show via a direct counterexample, however, that their progressive sampling algorithm using regret pruning fails to satisfy this second inclusion. Rather, their algorithm only yields the guarantee that all \emph{true} pure Nash equilibria of the empirical game are approximate equilibria of the simulation-based game, and is unable to directly yield any non-trivial guarantee regarding \emph{approximate} pure Nash equilibria of the empirical game. This difference is crucial, since in many games, computing a pure Nash equilibrium is computationally intractable (e.g., it is NP-complete in graphical games \citep{gottlob2005purenash}), and sometimes such an equilibrium does not even exist. In such cases, an approximate pure Nash equilibrium is the best that can be hoped for, but \citeauthor{areyan2020improved}'s progressive sampling algorithm using regret pruning yields no guarantees for such equilibria.

In response to this limitation of the original regret pruning technique, we design three new variations of regret pruning (and new corresponding progressive sampling algorithms). At the cost of a slightly tighter regret pruning criterion, our first regret pruning variation yields the guarantee regarding approximate equilibria of the empirical game which \citeauthor{areyan2020improved}'s regret pruning technique was originally intended to yield. Our second variation incorporates the non-uniform utility deviation bounds used by \citet{cousins2022computational} for well-estimated pruning in order to design a looser regret pruning criterion which nonetheless yields the same guarantees as the first variation. Finally, our third variation also takes advantage of non-uniform bounds, but uses different proof techniques than those used in the first two variations to yield both pure \emph{and} mixed Nash containment guarantees. 

The third regret pruning variation is a particularly significant contribution, as it is one of the first pruning techniques beyond simple well-estimated pruning for learning mixed equilibria. The only available alternative is rationalizability pruning introduced by \citet{areyan2020improved}, which requires the use of a computationally expensive iterative dominance algorithm, has a very tight pruning criterion which often prunes few to no strategy profiles in practice, and most importantly, like the original regret pruning technique, only yields guarantees regarding \emph{true} mixed Nash equilibria of the empirical game. 
In contrast, our novel third regret pruning variation utilizes a pruning criterion which is very cheap to compute, can prune a very significant number of simulation queries (which we confirm experimentally), and yields Nash containment guarantees for approximate mixed equilibria of the empirical game.

In addition to presenting the guarantees progressive sampling algorithms using these novel variations of regret pruning can satisfy, we also derive sample complexity bounds for these new algorithms. In particular, we present PAC-style upper bounds on the number of samples our progressive sampling algorithms will take to prune each respective strategy profile. Finally, we conclude by demonstrating experimentally that our novel progressive sampling algorithms which incorporate both well-estimated pruning and novel regret pruning variations significantly outperform \citeauthor{cousins2022computational}'s progressive sampling algorithm which used well-estimated pruning alone, in some cases requiring up to 50\% fewer simulation queries to learn equilibria of similar quality.

\subsection{Related Works}

The EGTA literature, while relatively young, is growing rapidly, with researchers actively contributing methods for myriad game models.
Some of these methods are designed for normal-form games~\citep{cousins2022computational,areyan2020improved,areyan2019learning,tavares2016rock,fearnley2015learning,vorobeychik2008stochastic}, and others, for extensive-form games~\citep{learningmaximin,Gattieqapprox,Zhang_Sandholm_2021}.
Most methods apply to games with finite strategy spaces, but some apply to games with infinite strategy spaces \citep{learningmaximin,vorobeychik2007learning,wiedenbeck2018regression}.
A related line of work aims to empirically design mechanisms via EGTA methodologies~\citep{vorobeychik2006empirical,DBLP:conf/uai/ViqueiraCMG19}.

The progressive sampling algorithms and pruning techniques which we design in this paper extend work done by \citet{cousins2022computational,areyan2020improved, areyan2019learning} in designing algorithms for learning equilibria in normal-form simulation-based games with finite strategy spaces.

\section{Learning Framework}

We begin with the standard definition of standard normal-form games, and some related properties
.
We then introduce our formal model of simulation-based games and empirical games.
Finally, we state the concentration inequalities we use to guide pruning in our progressive sampling algorithms.

\subsection{Basic Game Theory}

\begin{definition}[Normal-Form Game]
A \mydef{normal-form game} $\GameTuple \doteq \langle \SetOfPlayers, \{ \StrategySet_\PlayerIndex \}_{\PlayerIndex \in \SetOfPlayers} ,
\Utility \rangle$ consists of a set of players $\SetOfPlayers$,
each with a corresponding \mydef{pure strategy set} $\StrategySet_\PlayerIndex$.
We define $\StratProfileSpace \doteq \StrategySet_1 \times \dots \times \StrategySet_{\NumberOfPlayers}$ to be the \mydef{pure strategy profile space}, and then $\Utility : \StratProfileSpace \to \R^{\NumberOfPlayers}$ is a vector-valued \mydef{utility function} (equivalently, a vector of $\NumberOfPlayers$ scalar utility functions $\Utility_\PlayerIndex$).  
\end{definition}

Given an NFG $\GameTuple$, we denote by $\StrategySet_\PlayerIndex^\diamond$ the set of distributions over $\StrategySet_\PlayerIndex$; this set is called player $\PlayerIndex$'s \mydef{mixed strategy set}. We define $\MixedStratProfileSpace = \MixedStrategySet_1 \times \dots \times \MixedStrategySet_{\NumberOfPlayers}$ to be the \mydef{mixed strategy profile space}, and, overloading notation, we write $\Utility(\StratProfile)$
to denote the expected utility of a mixed strategy profile $\StratProfile \in \MixedStratProfileSpace$.
Each pure strategy profile $\StratProfile\in\StratProfileSpace$ is contained in the mixed strategy profile space $\MixedStratProfileSpace$, represented by the profile with each mixed strategy concentrated entirely at the respective pure strategy.

Given player $\PlayerIndex$ and strategy profile $\StratProfile\in\MixedStratProfileSpace$, the set $\Adjacent_{\PlayerIndex, \StratProfile} \doteq \{ (\StratProfile_1, \dots, \StratProfile_{\PlayerIndex - 1}, \StrategyAlt, \StratProfile_{p+1}, \dots, \StratProfile_{\abs{\SetOfPlayers}}) \mid \StrategyAlt\in\StrategySet_{\PlayerIndex}\}$ contains all adjacent strategy profiles, meaning those in which the strategies of all players $q \ne \PlayerIndex$ are fixed at $\StratProfile_q$,  while  player $\PlayerIndex$'s strategy may vary across their pure strategy set.

\begin{definition}[Regret]
    A player $\PlayerIndex$'s regret at strategy profile $\StratProfile\in\MixedStratProfileSpace$ is defined as $\displaystyle\Regret_\PlayerIndex(\StratProfile; \Utility) \doteq \sup_{\StratProfile' \in \Adjacent_{\PlayerIndex, \StratProfile}} \Utility_{\PlayerIndex} (\StratProfile') - \Utility_{\PlayerIndex} (\StratProfile)$.  
    We further define $\Regret(\StratProfile; \Utility)\doteq \max_{\PlayerIndex\in\SetOfPlayers}\Regret_\PlayerIndex(\StratProfile; \Utility)$.
\end{definition}

A strategy profile $\StratProfile\in\StratProfileSpace$ is player $\PlayerIndex$'s \mydef{best response} if $\Regret_\PlayerIndex(\StratProfile; \Utility) = 0$ (i.e., the player does not regret choosing this strategy profile as opposed to an adjacent one). We say it is an $\epsilon$-\mydef{best response} if $\Regret_\PlayerIndex(\StratProfile; \Utility) \leq \epsilon$.

A strategy profile $\StratProfile\in\MixedStratProfileSpace$ is an $\epsilon$-\mydef{Nash equilibrium} if it is an $\epsilon$-best response for each player $\PlayerIndex\in\SetOfPlayers$ (i.e., if $\Regret(\StratProfile; \Utility) \leq \epsilon$). If $\StratProfile$ corresponds to a pure strategy profile, then we call it an $\epsilon$-\mydef{pure Nash equilibrium} ($\epsilon$-PNE); otherwise, we call it an $\epsilon$-\mydef{mixed Nash equilibrium} ($\epsilon$-MNE).
A $0$-PNE is simply called a PNE, and a $0$-MNE is called an MNE.
The set of $\epsilon$-pure (resp. mixed) Nash equilibria is denoted $\Equilibria_\epsilon(\Utility)$ (resp. $\MixedEquilibria_\epsilon(\Utility)$), and the set of pure (resp. mixed) Nash equilibria is denoted $\Equilibria(\Utility)$ (resp. $\MixedEquilibria(\Utility)$).

\subsection{Formal Model of Simulation Based Games}

In simulation-based games, we assume access to a \mydef{simulator} $\Simulator (\cdot)$, which can be queried to produce unbiased noisy samples of the players' utilities when $\StratProfile \in \StratProfileSpace$ is played.
We denote such a sample by $\SampleUtility (\StratProfile) \sim \Simulator (\StratProfile)$, where $\SampleUtility (\StratProfile)$ is a $\abs{\SetOfPlayers}$-vector comprising utilities for each player.
%We now present our formal model for a simulation-based game.

\begin{definition}[Simulation-Based Game]
    A simulation-based game $\GameTuple_\Simulator \doteq \langle \SetOfPlayers, \StratProfileSpace, \Simulator \rangle$ consists of a set of players $\SetOfPlayers$, a pure strategy profile space $\StratProfileSpace\doteq \StrategySet_1 \times \dots \times \StrategySet_{\abs{\SetOfPlayers}}$, and a simulator $\Simulator(\cdot)$ that produces noisy samples $\SampleUtility (\StratProfile) \sim \Simulator (\StratProfile)$ upon simulation of a strategy profile $\StratProfile \in \StratProfileSpace$.
\end{definition}

Corresponding to each simulation-based game $\GameTuple_\Simulator$ is an expected normal-form game.

\begin{definition}[Expected Normal-Form Game]
    Given a simulation-based game $\GameTuple_\Simulator \doteq \langle \SetOfPlayers, \StratProfileSpace, \Simulator \rangle$, we define the ``underlying'' utility function $\SimUtility : \StratProfileSpace \to \mathbb{R}^{\abs{\SetOfPlayers}}$ by $\SimUtility(s) \doteq \Expect_{\SampleUtility(\StratProfile)\sim \Simulator (\StratProfile)}\left[\SampleUtility(\StratProfile)\right]$. 
    The expected game corresponding to $\GameTuple_\Simulator$ is then the normal-form game $\langle \SetOfPlayers, \StratProfileSpace, \SimUtility\rangle$.
    Overloading notation, we also let $\GameTuple_\Simulator$ denote this (unknown) expected normal-form game.
\end{definition}

Since we do not have direct access to the expected normal-form game, its utilities must be learned by repeatedly querying the simulator.
The resulting empirical estimate of the expected game is called an empirical game.

\begin{definition}[Empirical Normal-Form Game]
    Given a simulation-based game $\GameTuple_\Simulator \doteq \langle \SetOfPlayers, \StratProfileSpace, \Simulator\rangle$, let $\SampleUtility^{(1)}(\StratProfile)$, $\dots$, $\SampleUtility^{(\NumberQueries_{\StratProfile})}(\StratProfile)\sim \Simulator(\StratProfile)$ denote the sample utilities produced by $\NumberQueries_{\StratProfile} > 0$ queries to the simulator at strategy profile $\StratProfile \in \StratProfileSpace$. 
    We define the empirical utility function $\EUtility : \StratProfileSpace \to \mathbb{R}^{\NumberOfPlayers}$ by $\EUtility(\StratProfile)\doteq \frac{1}{\NumberQueries_{\StratProfile}}\sum_{i=1}^{\NumberQueries_{\StratProfile}}\SampleUtility^{(i)}(\StratProfile)$ for all $\StratProfile\in\StratProfileSpace$, and the ensuing empirical normal-form game by $\EmpiricalGame{\Simulator}\doteq \langle \SetOfPlayers, \StratProfileSpace, \EUtility\rangle$.
\end{definition}

From here onwards, let $\GameTuple_\Simulator\doteq \langle\SetOfPlayers, \StratProfileSpace, \Simulator\rangle$ be an arbitrary simulation-based game with underlying utility function $\Utility$, and let $\EmpiricalGame{\Simulator}\doteq\langle \SetOfPlayers, \StratProfileSpace, \EUtility\rangle$ be a corresponding empirical game. Using our formalization, we can now present one of the foundational results of EGTA \citep{tuyls2020bounds}.
\begin{restatable}{lemma}{tuylsDual}
    \label{thm:dual-containment-tuyls}
    If $\abs{\Utility_\PlayerIndex(\StratProfile) - \EUtility_\PlayerIndex(\StratProfile)}\leq \epsilon$ for all $(\PlayerIndex, \StratProfile)\in\UtilityIndices$, then
    \begin{equation*}
        \Equilibria(\Utility)\subseteq \Equilibria_{2\epsilon}(\EUtility) \subseteq \Equilibria_{4\epsilon}(\Utility)\textrm{ and }\MixedEquilibria(\Utility)\subseteq \MixedEquilibria_{2\epsilon}(\EUtility) \subseteq \MixedEquilibria_{4\epsilon}(\Utility)\enspace,
    \end{equation*}
    or more generally, $\Equilibria_{\gamma}(\Utility)\subseteq \Equilibria_{2\epsilon + \gamma}(\EUtility)$ and $\Equilibria_{\gamma}(\EUtility)\subseteq \Equilibria_{2\epsilon + \gamma}(\Utility)$ for all $\gamma \geq 0$ (resp. for mixed equilibria).
\end{restatable}
This result can be understood as stating that given a sufficiently strong approximation $\EmpiricalGame{\Simulator}$ of a simulation-based game $\GameTuple_\Simulator$, we can approximate pure (resp. mixed) Nash equilibria in $\GameTuple_\Simulator$ with perfect recall -- all pure (resp. mixed) Nash equilibria in $\GameTuple_{\Simulator}$ are approximate Nash equilibria in $\EmpiricalGame{\Simulator}$ -- and with \emph{approximately} perfect precision -- all approximate pure (resp. mixed) Nash equilibria in $\EmpiricalGame{\Simulator}$ are approximate pure (resp. mixed) Nash equilibria in $\GameTuple_{\Simulator}$. 
\Cref{thm:dual-containment-tuyls} is one of the primary motiviations for EGTA's pursuit of designing efficient algorithms for learning strong approximations of simulation-based games, as it guarantees that the better an approximation an empirical game is of a simulation-based game, the more strategically representative the empirical game will be of the underlying simulation-based game.

\subsection{Tail Bounds}

Next, we state the tail bounds upon which our novel regret pruning techniques and progressive sampling algorithms depend. These are the same bounds derived and used by \citeauthor{cousins2022computational}; for a more thorough discussion of them see \citet{cousins2022computational}. For all subsequent results, we make the following ``bounded utilities'' assumption.

\begin{assumption}[Bounded Utilities]
    For each strategy profile $\StratProfile\in\StratProfileSpace$, the sample utilities produced via $\Simulator(\StratProfile)$ lie on the bounded interval $[a_{\StratProfile}, b_{\StratProfile}]$ for some fixed $a_{\StratProfile}, b_{\StratProfile}\in \mathbb{R}$. We define $c := \sup_{\StratProfile\in\StratProfileSpace}(b_{\StratProfile} - a_{\StratProfile})$.
\end{assumption}

The most straight-forward tail bound for mean-estimation is Hoeffding's Inequality, which was used by \citet{tuyls2020bounds}. We use Hoeffding's inequality to bound each individual utility, combined with a union bound to yield a guarantee for all utilities.

\begin{theorem}[Hoeffding's Inequality]
    \label{thm:hoeffding}
    Let $\EmpiricalGame{\Simulator}\doteq\langle \SetOfPlayers, \StratProfileSpace, \EUtility\rangle$ be an empirical game. Then, with probability at least $1 - \delta$, for all $(\PlayerIndex, \StratProfile)\in \UtilityIndices$, it holds that
    \begin{equation*}
        \abs{\Utility_p(\StratProfile) - \EUtility_p(\StratProfile)} \leq \UtilityRange\sqrt{\frac{\ln\left(\nicefrac{2\abs{\UtilityIndices}}{\delta}\right)}{2\NumberQueries_{\StratProfile}}}\doteq \epsilonH_\PlayerIndex(\StratProfile)\enspace.
    \end{equation*}
\end{theorem}

Let $\UtilityVariance_{\PlayerIndex}(\StratProfile)$ denote $ \Var_{\SampleUtility(\StratProfile)\sim \Simulator (\StratProfile)}\left[\SampleUtility_{\PlayerIndex}(\StratProfile)\right]$ for all $(\PlayerIndex, \StratProfile) \in \SetOfPlayers \times \StratProfileSpace$. Since Hoeffding's Inequality assumes a worst-case variance on the utilities (i.e., $\UtilityVariance_{\PlayerIndex}(\StratProfile) = \nicefrac{\UtilityRange^2}{4}$), when variances are small, it yields a very loose bound. When the variances of utilities are known, Bennett's inequality provides a non-uniform, variance-sensitive guarantee.

\begin{theorem}[Bennett's Inequality]
    \label{thm:bennett}
    Let $\EmpiricalGame{\Simulator}\doteq\langle \SetOfPlayers, \StratProfileSpace, \EUtility\rangle$ be an empirical game. Then, with probability at least $1 - \delta$, for all $(\PlayerIndex, \StratProfile)\in \UtilityIndices$, it holds that
    \begin{equation*}
        \abs{\Utility_p(\StratProfile) - \EUtility_p(\StratProfile)} \leq \frac{\UtilityRange \ln\left(\nicefrac{2\abs{\UtilityIndices}}{\delta}\right)}{3\NumberQueries_{\StratProfile}} + \sqrt{\frac{2\UtilityVariance_{\PlayerIndex}(\StratProfile)\ln\left(\nicefrac{2\abs{\UtilityIndices}}{\delta}\right)}{\NumberQueries_{\StratProfile}}}
    \end{equation*}
\end{theorem}

Of course, since our only access to the utilities of the simulation-based game are via the simulator, we do not know their variances. \citet{cousins2022computational} circumvent this limitation by deriving an ``empirical Bennett's inequality'' depending on empirical estimates of the true utility variances.

\begin{theorem}[Empirical Bennett's Inequality]
    \label{thm:eBennett}
    Let $\EmpiricalGame{\Simulator}\doteq\langle \SetOfPlayers, \StratProfileSpace, \EUtility\rangle$ be an empirical game. Let $\kappa_\delta\doteq\left(\frac{1}{3} + \frac{1}{2\ln\left(\frac{3\abs{\UtilityIndices}}{\delta}\right)}\right)$. For all $(\PlayerIndex, \StratProfile)\in \UtilityIndices$, define
    \begin{align*}
        \footnotesize{\EUtilityVariance_{\PlayerIndex}(\StratProfile)} 
        &
        \doteq 
        \mathsmaller{\frac{1}{\NumberOfSamples - 1}} 
        \ssum_{j = 1}^{\NumberOfSamples} 
        \left( \Utility_{\PlayerIndex}(\StratProfile; \SamplePoint_j) -
        \hat{\Utility}_{\PlayerIndex}(\StratProfile; \Samples) \right)^2;
        \\
        \EVarianceEpsilon_{\PlayerIndex}(\StratProfile)
        &\doteq \mathsmaller{\frac{2 \UtilityRange^{2} \ln \left( \frac{3 \abs{\UtilityIndices}}{\delta} \right)}{3 \NumberOfSamples} + \sqrt{\kappa_\delta \Bigl( \frac{\UtilityRange^{2} \ln \left( \frac{3\abs{\UtilityIndices}}{\delta} \right)}{\NumberOfSamples - 1}\Bigr)^{\smash{2}} + \frac{2 \UtilityRange^{2} \EUtilityVariance_\PlayerIndex(\StratProfile) \ln \left( \frac{3\abs{\UtilityIndices}}{\delta} \right)}{\NumberOfSamples}}};\\
        \epsilonEB_\PlayerIndex(\StratProfile)
        &\doteq \mathsmaller{\frac{\UtilityRange\ln \left( \frac{3 \abs{\UtilityIndices}}{\delta} \right)}{3 \NumberOfSamples} + \sqrt{\frac{2 \left(\EUtilityVariance_\PlayerIndex (\StratProfile) + \EVarianceEpsilon_{\PlayerIndex}( \StratProfile)\right) \ln \left( \frac{3\abs{\UtilityIndices}}{\delta} \right)}{\NumberOfSamples}}}.
    \end{align*}
    Then, with probability at least $1 - \delta$, for all $(\PlayerIndex, \StratProfile)\in \UtilityIndices$, it holds that $\abs{\Utility_p(\StratProfile) - \EUtility_p(\StratProfile)} \leq \epsilonEB_{\PlayerIndex}(\StratProfile)$.
\end{theorem}

This empirical Bennett guarantee forms the basis for our progressive sampling algorithms.

\subsection{Progressive Sampling Algorithms}
\begin{algorithm*}%[htbp]
\algrenewcommand\algorithmicindent{2.0em}
\begin{algorithmic}[1]
\Procedure{PSP}{$\ConditionalGame{\ConditionSpace}, \ConditionDistribution, \UtilityIndices, \UtilityRange, \delta, \epsilon$}
%$\to (\tilde{\bm{\Utility}}, \tilde{\epsilon})$
 
\State \parbox[t]{\dimexpr\textwidth-\leftmargin-\labelsep-\labelwidth}{%
\Input 
Conditional game $\ConditionalGame{\ConditionSpace}$,
condition distribution $\ConditionDistribution$,
index set $\UtilityIndices$,
%sampling schedule $\SamplingSchedule$,
% utility range $\UtilityRange$,
failure probability $\delta \in (0, 1)$,
target error $\epsilon > 0$\strut}

% \State \Output 
% Empirical utilities $\tilde{\Utility}$,
% %and errors $\tilde{\bm{\epsilon}}$, 
% for all indices $(\PlayerIndex, \StratProfile) \in \SetOfPlayers \times \StratProfileSpace$

\State{Initialize empirical utilities $\EUtility^{(0)}_\PlayerIndex(\StratProfile) = 0$ for $(\PlayerIndex, \StratProfile)\in\UtilityIndices$}

\State{Initialize utility deviation bounds $\hat{\epsilon}^{(0)}_\PlayerIndex(\StratProfile) = \infty$ for $(\PlayerIndex, \StratProfile)\in\UtilityIndices$}

\State{Initialize active utility index set $\UtilityIndices^{(0)}\gets \UtilityIndices$}

\State{Initialize a sampling schedule $\NumberOfSamples_1, \dots, \NumberOfSamples_\ScheduleLength$ and cumulative sample size $\CumulativeSamples_0\gets 0$}

\For{$\TimeIndex \in 1, \dots, \ScheduleLength$}
    
    \For{$\StratProfile\in\StratProfileSpace$}
        \State{Determine unpruned player indices $\SetOfPlayers(\StratProfile; \UtilityIndices^{(\TimeIndex)}) \gets \{\PlayerIndex\in\SetOfPlayers \mid (\PlayerIndex, \StratProfile)\in \UtilityIndices^{(\TimeIndex)}\}$ at strategy profile $\StratProfile$}

        \State{Query simulator for utilities of unpruned players: $\{\SampleUtility_\PlayerIndex(\StratProfile) \mid \PlayerIndex\in\SetOfPlayers(\StratProfile; \UtilityIndices^{(\TimeIndex)})\}\gets \Simulator(\StratProfile; \SetOfPlayers(\StratProfile; \UtilityIndices^{(\TimeIndex)}))$}
    
        \State{Update empirical utilities $\EUtility^{(t)}_\PlayerIndex(\StratProfile) \gets \frac{\CumulativeSamples_{\TimeIndex - 1}}{\CumulativeSamples_{\TimeIndex - 1} + \NumberOfSamples_\TimeIndex}\cdot \EUtility^{(t-1)}_\PlayerIndex(\StratProfile) + \frac{\NumberOfSamples_\TimeIndex}{\CumulativeSamples_{\TimeIndex - 1} + \NumberOfSamples_\TimeIndex}\cdot \SampleUtility_\PlayerIndex(\StratProfile)$ for $\PlayerIndex\in\SetOfPlayers(\StratProfile; \UtilityIndices^{(\TimeIndex)})$}

        \State{Compute new utility deviation bounds $\hat{\epsilon}^{(t)}_\PlayerIndex(\StratProfile)$ for $\PlayerIndex\in\SetOfPlayers(\StratProfile; \UtilityIndices^{(\TimeIndex)})$, each with failure probability $\nicefrac{\delta}{\abs{\UtilityIndices}\ScheduleLength}$}
    
    \EndFor

    \State{Update cumulative sample size $\CumulativeSamples_\TimeIndex \gets \CumulativeSamples_{\TimeIndex - 1} + \NumberOfSamples_\TimeIndex$}

    \State{\parbox[t]{\dimexpr\textwidth-\leftmargin-\labelsep-\labelwidth}{%
    Prune any indices in $\UtilityIndices^{(\TimeIndex)}$ which do not require further estimation
    (e.g., well-estimated pruning:\\ $\UtilityIndices^{(\TimeIndex)} \gets \smash{\{ (\PlayerIndex, \StratProfile) \in \UtilityIndices^{(\TimeIndex - 1)} \mid \hat{\epsilon}^{(t)}_{\PlayerIndex} (\StratProfile) > \epsilon \}}$, i.e., prune indices that have met the target $\epsilon$ error guarantee)\strut}} \label{alg:psp:pruning}
    
    \If{all indices in $\UtilityIndices^{(\TimeIndex)}$ are pruned (i.e., $\UtilityIndices^{(\TimeIndex)} = \emptyset$)}
        
        \State \Return{empirical utilities $\EUtility^{(t)}$}
    \EndIf

\EndFor

%\State \Error ``\textsc{Insufficient samples to produce $\epsilon$-$\delta$ guarantee}''
\EndProcedure
\end{algorithmic}

\caption{General Progressive Sampling Algorithm}
\label{alg:psp}
\end{algorithm*}

Finally, we present the general class of progressive sampling (PS) algorithms (see \Cref{alg:psp}) for learning simualation-based games. As the name suggests, progressive sampling algorithms work by progressively sampling utilities, pruning those which are sufficiently estimated for the relevant learning goal at hand. One core component of progressive sampling algorithms is the sampling schedule. On each iteration, a progressive sampling algorithm will collect the number of samples dictated by the sampling schedule for each active (i.e., unpruned) utility index. It will then use the samples to update the empirical game and to compute new utility deviation bounds (which in our case will be dependent on the tail bounds introduced earlier). Notice that the utility deviation bounds must each have individual failure probability $\frac{\delta}{\abs{\UtilityIndices}\ScheduleLength}$, as opposed to just $\frac{\delta}{\abs{\UtilityIndices}}$ as in \Cref{thm:eBennett}. This is to ensure that, via an additional union bound, all pruned indices will, with high probability (w.h.p.), have been pruned justifiably with respect to the true game. Finally, the empirical game and utility deviation bounds will be used to inform which utility indices can be pruned on the current iteration. Progressive sampling algorithms terminate once either all utility indices are pruned, or the sampling schedule is exhausted.

In this paper, we focus on PS algorithms for learning equilibria. On the basis of \Cref{thm:dual-containment-tuyls}, one sufficient condition for pruning an index $(\PlayerIndex, \StratProfile)\in\UtilityIndices$ is that it is estimated (w.h.p.) to within some target error $\epsilon$. If this is the only pruning criteria used, then upon termination of the algorithm, if all indices have been pruned, the resulting empirical game will (w.h.p.) satisfy $\abs{\Utility_\PlayerIndex(\StratProfile) - \EUtility_\PlayerIndex(\StratProfile)}\leq \epsilon$ for all $(\PlayerIndex, \StratProfile)\in\UtilityIndices$, and will thus (w.h.p.) satisfy the pure and mixed dual Nash containment results in \Cref{thm:dual-containment-tuyls}. \citet{cousins2022computational} design precisely this algorithm, using Hoeffding (\Cref{thm:hoeffding}) and empirical Bennett (\Cref{thm:eBennett}) bounds to inform their pruning (pruning once the guarantee corresponding to an index is tighter than the target error $\epsilon$), and carefully crafting a sampling schedule which guarantees that all indices will be pruned prior to its exhaustion. They refer to this pruning approach as \emph{well-estimated pruning} \citep{cousins2022computational}. We show an example of how the pruning criteria may be implemented in \Cref{alg:psp} \Cref{alg:psp:pruning}.

\section{Regret Pruning}

Having covered the requisite background, we can now begin our discussion of regret pruning and present our novel regret-pruning variations. As discussed earlier, \Cref{thm:dual-containment-tuyls} immediately suggests well-estimated pruning as a pruning approach, and this approach was used to design the PS algorithm presented in \citet{cousins2022computational} (which we henceforth refer to as PS-WE for Progressive Sampling with Well-Estimated Pruning). 
\begin{restatable}{theorem}{algoPSWE}
    If \PSWE{}$(\ConditionalGame{\ConditionSpace}, \ConditionDistribution, \UtilityIndices, \delta, \epsilon)$ returns an empirical utility function $\EUtility$, then with probability at least $1 - \delta$, for all $\gamma\geq 0$, it holds that 
    \begin{enumerate}
        \item $\Equilibria_{\gamma}(\Utility)\subseteq \Equilibria_{2\epsilon + \gamma}(\EUtility)$ and $\Equilibria_{\gamma}(\EUtility)\subseteq \Equilibria_{2\epsilon + \gamma}(\Utility)$
        \item $\MixedEquilibria_{\gamma}(\Utility)\subseteq \MixedEquilibria_{2\epsilon + \gamma}(\EUtility)$ and $\MixedEquilibria_{\gamma}(\EUtility)\subseteq \MixedEquilibria_{2\epsilon + \gamma}(\Utility)$\enspace. 
    \end{enumerate}
\end{restatable}
The condition presented in \Cref{thm:dual-containment-tuyls} (i.e., $\abs{\Utility_\PlayerIndex(\StratProfile) - \EUtility_\PlayerIndex(\StratProfile)}\leq \epsilon$ for all $(\PlayerIndex, \StratProfile)\in\UtilityIndices$), however, is not the only condition to yield these kinds of Nash containment results. Consider the PS algorithm designed in \citet{areyan2020improved} (which we refer to as \PSRegOld{} for Progressive Sampling with Regret Pruning; the 0 will distinguish this algorithm from our new variations). Unlike \PSWE{}, \PSRegOld{} uses uniform utility deviation bounds -- they simply bound all utility deviations by $\sup_{(\PlayerIndex, \StratProfile)\in\UtilityIndices} \epsilonEB_\PlayerIndex(\StratProfile)$ ($\epsilonEB$ defined as in \Cref{thm:eBennett})\footnote{\citet{areyan2020improved} use a version of the empirical Bennett tail bounds that is a constant factor looser than the one presented here, as the tighter bound had not been derived until \citet{cousins2022computational}}. \PSRegOld{} uses well-estimated pruning (though the authors do not explicitly call it that), but it additionally uses what the authors call ``regret pruning'' \citep{areyan2020improved}. The authors claim that though using this pruning approach does \emph{not} guarantee that the condition in \Cref{thm:dual-containment-tuyls} is met upon termination of the algorithm, it nonetheless guarantees that (w.h.p.) the pure Nash containment result $\Equilibria(\Utility)\subseteq \Equilibria_{2\epsilon}(\EUtility)\subseteq \Equilibria_{4\epsilon}(\Utility)$ is satisfied. In this section, we present a counterexample which shows that their pruning approach does not in fact guarantee this pure Nash containment result. We show that, instead, their pruning approach is only able to guarantee a weaker Nash containment result. 
Furthermore, we also present 3 novel variations of the regret pruning criterion presented in \citet{areyan2020improved} which each have varying benefits and costs in comparison to the original. 
The first variation is a generalization of the original regret pruning criterion with respect to a hyper-parameter $\gamma^*\geq 0$. We show that when $\gamma^* = 0$, this variation is identical to the regret pruning criterion from \PSRegOld{}. When $\gamma^* = 2\epsilon$, however, we show that, at the cost of taking slightly longer to prune indices, this new variation yields the stronger pure Nash containment guarantee which \citet{areyan2020improved} originally intended their pruning criterion to meet (i.e., $\Equilibria(\Utility)\subseteq \Equilibria_{2\epsilon}(\EUtility)\subseteq \Equilibria_{4\epsilon}(\Utility)$). The second variation takes advantage of non-uniform utility deviation bounds to yield the same guarantee as the first variation, while pruning indices significantly sooner in practice than otherwise. Finally, the third variation modifies the second variation, yielding a mixed Nash containment guarantee in addition to the same pure Nash guarantee, at the cost of a slightly tighter pruning criterion than variation 2.

\subsection{Old Regret Pruning}

We begin by presenting the regret pruning criterion used in \PSRegOld{} from \citet{areyan2020improved}. Since \PSRegOld{} uses uniform utility deviation bounds we simply use $\hat{\epsilon}^{(t)}$ to denote the utility deviation bound on $\EUtility^{(t)}$ at iteration $t$. \PSRegOld{} prunes an index $(\PlayerIndex, \StratProfile)\in\UtilityIndices$ on an iteration $t$ if any of the following holds:
\begin{enumerate}
    \item $\hat{\epsilon}^{(t)}\leq \epsilon$ (well-estimated pruning)
    \item $\Regret_\PlayerIndex(\StratProfile; \EUtility^{(t)}) \geq 2\hat{\epsilon}^{(t)}$ (regret pruning).
\end{enumerate}
\citet{areyan2020improved} claim that \PSRegOld{} satisfies the following guarantee.
\begin{claim}
    \label{clm:old-regret}
     If \PSRegOld{}$(\ConditionalGame{\ConditionSpace}, \ConditionDistribution, \UtilityIndices, \delta, \epsilon)$ returns an empirical utility function $\EUtility$, then with probability at least $1 - \delta$, it holds that $\Equilibria(\Utility)\subseteq \Equilibria_{2\epsilon}(\EUtility)\subseteq \Equilibria_{4\epsilon}(\Utility)$.
\end{claim}

We present a counter-example that shows that the second inclusion $\Equilibria_{2\epsilon}(\EUtility)\subseteq \Equilibria_{4\epsilon}(\Utility)$ does not necessarily hold.

\begin{counterexample}
    Consider a two-player game $\GameTuple$ in which player $A$ has two pure strategies, $a_1$ and $a_2$, while player $B$ has just one pure strategy $b$. Define player $A$'s utility function by $\Utility_A(a_1, b) = 2$ and $\Utility_A(a_2, b) = 1$. Suppose we run \PSRegOld{} with target error $\epsilon\doteq 0.2$ and get the following:
    \begin{alignat*}{3}
        &\textup{Iteration 1:}\quad
        \boxed{
        \begin{aligned}
            &\EUtility^{(1)}_A(a_1, b) = 2.5\\
            &\EUtility^{(1)}_A(a_2, b) = 1.45\\
            &\hat{\epsilon}^{(1)}=0.5\\
        \end{aligned}}&\\
        &\quad\Longrightarrow\quad
        \begin{aligned}[t]
            &\textup{Index $(A, (a_2, b))$ regret pruned}\\
            &\textup{(since $\hat{\epsilon}^{(1)} = 0.5 < \frac{\Regret_A((a_2, b); \EUtility^{(1)})}{2} = 0.525$)}
        \end{aligned}\\
        &\textup{Iteration 2:}\quad
        \boxed{
        \begin{aligned}
            &\EUtility^{(2)}_A(a_1, b) = 1.8\\
            &\EUtility^{(2)}_A(a_2, b) = 1.45\\
            &\hat{\epsilon}^{(2)}=0.2\\
        \end{aligned}}&\\\\
        &\quad\Longrightarrow\quad
        \begin{aligned}[t]
        &\textup{Index $(A, (a_1, b))$ well-estimated pruned;}\\
        &\textup{\PSRegOld{} terminates with $\EUtility_A=\EUtility^{(2)}_A$}
        \end{aligned}
    \end{alignat*}
    We have that $(a_2, b)\in \Equilibria_{2\epsilon}(\EUtility)$, since $\Regret_A((a_2, b); \EUtility) = 1.8 - 1.45 = 0.35 < 2\epsilon = 0.4$ and $\Regret_B((a_2, b); \EUtility) = 0$. But we have $(a_2, b)\not\in \Equilibria_{4\epsilon}(\Utility)$, since $\Regret_A((a_2, b); \Utility) = 2 - 1 = 1 > 4\epsilon = 0.8$. Hence, we have $\Equilibria_{2\epsilon}(\EUtility)\not\subseteq \Equilibria_{4\epsilon}(\Utility)$. Since all utility deviation guarantees have been held, this is not a failure case. Thus, \Cref{clm:old-regret} cannot be true.
\end{counterexample}

Though the second inclusion in \Cref{clm:old-regret} cannot be guaranteed, we show that an alternative inclusion does hold (w.h.p.).

\begin{restatable}{theorem}{regretOld}
    If \PSRegOld{}$(\ConditionalGame{\ConditionSpace}, \ConditionDistribution, \UtilityIndices, \delta, \epsilon)$ returns an empirical utility function $\EUtility$,
    then with probability at least $1 - \delta$, we have that $\Equilibria(\Utility)\subseteq \Equilibria_{2\epsilon}(\EUtility)$ and $\Equilibria(\EUtility)\subseteq \Equilibria_{2\epsilon}(\Utility)$.
\end{restatable}

From this guarantee, we see that if a game analyst is using \PSRegOld{} to learn an approximate pure Nash equilibrium of the simulation-based game, they will need to compute a (true) pure Nash equilibrium of the resulting empirical game. Of course, in many games, computing a pure Nash equilibrium is computationally intractable (e.g., it is NP-complete in graphical games \citep{gottlob2005purenash}); the best that can be hoped for is an approximate pure Nash equilibrium. Furthermore, even if the simulation-based game has a pure Nash equilibrium, it is not then guaranteed that the empirical game will also have a pure Nash equilibrium, but rather only that it will have a $2\epsilon$-pure Nash equilibrium. A limitation of \PSRegOld{} is that it is not able to provide any guarantees regarding the efficacy of approximate pure Nash equilibria from the empirical game in the true game.

\subsection{New Regret Pruning Variations}

We now introduce our novel regret pruning criteria. We begin by presenting a variation which resolves the limitation observed in \PSRegOld{} of lacking guarantees regarding empirical approximate Nash equilibria. This variation is derived on the basis of a stronger version of \Cref{thm:dual-containment-tuyls}.

\begin{restatable}{lemma}{dualNew}
    \label{lem:new-dual-containment}
    Let $\gamma^*\geq 0$. If $\abs{\Utility_\PlayerIndex(\StratProfile) - \EUtility_\PlayerIndex(\StratProfile)}\leq \epsilon$ for all $(\PlayerIndex, \StratProfile)\in\UtilityIndices$ satisfying $\Regret_\PlayerIndex(\StratProfile; \Utility) = 0$ or $\Regret_\PlayerIndex(\StratProfile; \EUtility) \leq \gamma^*$, then $\Equilibria(\Utility)\subseteq \Equilibria_{2\epsilon}(\EUtility)$ and $\Equilibria_{\gamma}(\EUtility)\subseteq \Equilibria_{2\epsilon + \gamma}(\Utility)$ for all $0\leq \gamma \leq \gamma^*$.
\end{restatable}

Whereas \Cref{thm:dual-containment-tuyls} required all indices to be well-estimated, \Cref{lem:new-dual-containment} requires only indices with sufficiently low regret in both the true game and empirical game to be well-estimated. This gives room for indices with provably high regret to be regret-pruned. Of course, as \citet{areyan2020improved} also observed, this potential for regret pruning seems to come at the cost of any guarantees regarding mixed Nash equilibria. On the basis of \Cref{lem:new-dual-containment}, we design our first regret pruning variation.

\begin{restatable}{theorem}{regretU}
    \label{thm:regret-pruning-1}
    Consider a PS algorithm, \PSRegU{}, using uniform utility deviation bounds, which conducts well-estimated pruning and on each iteration $t$, also \mydef{regret-prunes} any index $(\PlayerIndex, \StratProfile)\in\UtilityIndices$ which satisfies
    \begin{equation*}
        \Regret_\PlayerIndex(\StratProfile; \EUtility^{(t)}) > \max\{2\hat{\epsilon}^{(t)}, \gamma^* + \epsilon + \hat{\epsilon}^{(t)}\}\enspace.
    \end{equation*}
    If \PSRegU{}$(\ConditionalGame{\ConditionSpace}, \ConditionDistribution, \UtilityIndices, \delta, \epsilon, \gamma^*)$ returns an empirical utility function $\EUtility$, then with probability at least $1 - \delta$, it holds that $\Equilibria(\Utility)\subseteq \Equilibria_{2\epsilon}(\EUtility)$ and $\Equilibria_{\gamma}(\EUtility)\subseteq \Equilibria_{2\epsilon + \gamma}(\Utility)$ for all $0\leq \gamma\leq \gamma^*$.
\end{restatable}

Notice that when $\gamma^* = 0$, \PSRegU{} is identical to \PSRegOld{}, even yielding the same exact guarantees. \PSRegU{} is thus a generalization of \PSRegOld{} to cases where $\gamma^* > 0$. When $\gamma^* > 0$, \PSRegU{} yields, at the cost of potentially reduced pruning, a stronger guarantee upon termination than \PSRegOld{}, ensuring that even an approximate empirical pure Nash equilibrium (so long as it is at worst a $\gamma^*$-pure Nash equilibrium) will be an approximate pure Nash equilibrium in the true game. In practice, the parameter $\gamma^*$ can be set to the smallest value for which the game analyst is still certain they will be able to compute a $\gamma^*$-pure Nash equilibria of the empirical game. If we set $\gamma^* = 2\epsilon$, we get the dual pure Nash containment guarantee, $\Equilibria(\Utility)\subseteq \Equilibria_{2\epsilon}(\EUtility)\subseteq \Equilibria_{4\epsilon}(\Utility)$, which \PSRegOld{} was originally designed to meet.

One limitation of both \PSRegU{} and \PSRegOld{} is that they both depend only on uniform utility deviation bounds. The next algorithm and regret pruning variation takes advantage of non-uniform utility deviation bounds (\Cref{thm:eBennett}) to prune potentially more (and in practice significantly more; see \Cref{fig:query-one-game}) indices via both well-estimated pruning and regret-pruning, while yielding the same guarantees as \PSRegU{}. In the following result, we use $\Regret^\downarrow_\PlayerIndex(\StratProfile; \EUtility)$ to denote the high-probability lower-bound on $\Regret_\PlayerIndex(\StratProfile; \Utility)$ defined by
\begin{equation*}
    \Regret^\downarrow_\PlayerIndex(\StratProfile; \EUtility) \doteq \sup_{\StratProfile' \in \Adjacent_{\PlayerIndex, \StratProfile}} (\Utility_{\PlayerIndex} (\StratProfile') - \epsilon_\PlayerIndex(\StratProfile')) - (\Utility_{\PlayerIndex} (\StratProfile) + \epsilon_\PlayerIndex(\StratProfile)).
\end{equation*}
\begin{restatable}{theorem}{regretNU}
    \label{thm:regret-pruning-2}
    Consider a PS algorithm, \PSRegNU{}, using non-uniform utility deviation bounds, which conducts well-estimated pruning and on each iteration $t$, also \mydef{regret-prunes} any index $(\PlayerIndex, \StratProfile)\in\UtilityIndices$ which satisfies
    \begin{equation*}
        \Regret^\downarrow_\PlayerIndex(\StratProfile; \EUtility^{(t)}) > \max\{0, \gamma^* + \epsilon - \hat{\epsilon}^{(t)}_\PlayerIndex(\StratProfile)\}\enspace.
    \end{equation*}
    If \PSRegNU{}$(\ConditionalGame{\ConditionSpace}, \ConditionDistribution, \UtilityIndices, \delta, \epsilon, \gamma^*)$ returns an empirical utility function $\EUtility$, then with probability at least $1 - \delta$, it holds that $\Equilibria(\Utility)\subseteq \Equilibria_{2\epsilon}(\EUtility)$ and $\Equilibria_{\gamma}(\EUtility)\subseteq \Equilibria_{2\epsilon + \gamma}(\Utility)$ for all $0\leq \gamma \leq \gamma^*$.
\end{restatable}
% If the utility deviation bounds used in \PSRegNU{} end up coincidentally being uniform, then \Cref{thm:regret-pruning-2} matches \Cref{thm:regret-pruning-1} exactly.\bhaskar{Might not need to say this}

Notice again that when $\gamma^* = 0$, the pruning criterion becomes simply $\Regret_\PlayerIndex^\downarrow(\StratProfile; \EUtility^{(\TimeIndex)}) > 0$ and yields the same guarantees as \PSRegOld{}. Thus, when $\gamma^* = 0$, \PSRegNU{} can also reasonably be called \PSRegOld{}+.

Since all the aforementioned regret pruning variations derive from the result presented in \Cref{lem:new-dual-containment}, they are only able to provide guarantees regarding pure equilibria. This is the most glaring limitation of the new PS algorithms presented so far, since many games do not even have strong approximate pure Nash equilibria which the algorithms could potentially be used to learn. In contrast, \PSWE{} derives from \Cref{thm:dual-containment-tuyls}, and thus yields mixed Nash containment guarantees, but of course does not allow for regret-pruning. We now present a lemma which serves as a middle ground between \Cref{lem:new-dual-containment} and \Cref{thm:dual-containment-tuyls}.

\begin{restatable}{lemma}{dualMiddle}
    \label{lem:middle-dual-containment}
    If $\abs{\Utility_\PlayerIndex(\StratProfile) - \EUtility_\PlayerIndex(\StratProfile)}\leq \max\left\{\epsilon, \frac{\Regret_\PlayerIndex(\StratProfile; \EUtility)}{2}\right\}$ for all $(\PlayerIndex, \StratProfile)\in\UtilityIndices$, then for all $0\leq \gamma \leq 2\epsilon$, it holds that
    \begin{equation*}
        \Equilibria(\Utility)\subseteq \Equilibria_{2\epsilon}(\EUtility) 
        \textup{ and }
        \Equilibria_{\gamma}(\EUtility)\subseteq \Equilibria_{2\epsilon + \gamma}(\Utility),
    \end{equation*}
    and for all $\gamma \geq 0$, it holds that
    \begin{equation*}
        \MixedEquilibria(\Utility)\subseteq \MixedEquilibria_{4\epsilon}(\EUtility)
        \textup{ and }
        \MixedEquilibria_{\gamma}(\EUtility)\subseteq \MixedEquilibria_{2\epsilon + \frac{3\gamma}{2}}(\Utility).
    \end{equation*}
\end{restatable}

The condition in \Cref{lem:middle-dual-containment} is looser than that in \Cref{thm:dual-containment-tuyls}, opening up the potential for regret-pruning, but is (strictly) tighter than the condition in \Cref{lem:new-dual-containment} when $\gamma^* = 2\epsilon$ (i.e., empirical games which satisfy the condition in \Cref{lem:middle-dual-containment} necessarily satisfy the condition in \Cref{lem:new-dual-containment} when $\gamma^* = 2\epsilon$, but the converse does not hold; see Appendix for proof), which allows for the additional mixed Nash containment guarantee. Notice further that unlike \Cref{lem:new-dual-containment}, \Cref{lem:middle-dual-containment} does not depend on any additional parameter $\gamma^*$ and yields guarantees for all $\gamma\geq 0$, rather than just $\gamma\in [0, \gamma^*]$. We use \Cref{lem:middle-dual-containment} to derive yet another regret pruning variation.

\begin{restatable}{theorem}{regretM}
    \label{thm:regret-pruning-3}
    Consider a PS algorithm, \PSRegM{}, using non-uniform utility deviation bounds, which conducts well-estimated pruning and on each iteration $t$, also \mydef{regret-prunes} any index $(\PlayerIndex, \StratProfile)\in\UtilityIndices$ which satisfies
    \begin{equation*}
        \Regret^\downarrow_\PlayerIndex(\StratProfile; \EUtility^{(t)}) > \epsilon + \hat{\epsilon}^{(t)}_\PlayerIndex(\StratProfile)\enspace.
    \end{equation*}
    If \PSRegM{}$(\ConditionalGame{\ConditionSpace}, \ConditionDistribution, \UtilityIndices, \delta, \epsilon)$ returns an empirical utility function $\EUtility$, then with probability at least $1 - \delta$, for all $0\leq \gamma \leq 2\epsilon$, it holds that
    \begin{equation*}
        \Equilibria(\Utility)\subseteq \Equilibria_{2\epsilon}(\EUtility) 
        \textup{ and }
        \Equilibria_{\gamma}(\EUtility)\subseteq \Equilibria_{2\epsilon + \gamma}(\Utility),
    \end{equation*}
    and for all $\gamma \geq 0$, it holds that
    \begin{equation*}
        \MixedEquilibria(\Utility)\subseteq \MixedEquilibria_{4\epsilon}(\EUtility)
        \textup{ and }
        \MixedEquilibria_{\gamma}(\EUtility)\subseteq \MixedEquilibria_{2\epsilon + \frac{3\gamma}{2}}(\Utility).
    \end{equation*}
\end{restatable}

The mixed Nash containment guarantee achieved by \PSRegM{} is a $\frac{\gamma}{2}$ factor looser than that achieved by \PSWE{} (\Cref{thm:dual-containment-tuyls}). As a result, slightly stronger approximate empirical mixed Nash equilibria will need to be computed when using \PSRegM{} in order to guarantee an equally strong approximate true mixed Nash equilibrium as in \PSWE{}.

\subsection{Efficiency Bounds and Correctness}

Using similar proof techniques to those used in \citet{cousins2022computational} to derive efficiency bounds for \PSWE{}, we derive upper bounds on the number of samples each utility index requires prior to being pruned by \PSRegNU{} and \PSRegM{}, respectively. In the following results, suppose that the utility deviation bounds being used by the mentioned PS algorithms are the minimum of Hoeffding bounds (\Cref{thm:hoeffding}) and empirical Bennett bounds (\Cref{thm:eBennett}).

\begin{restatable}{theorem}{efficiency}
    When running \PSRegNU{}$(\ConditionalGame{\ConditionSpace}, \ConditionDistribution, \UtilityIndices, \delta, \epsilon, \gamma^*)$, with probability at least $1 - \frac{\delta}{3}$, the index $(\PlayerIndex, \StratProfile)\in \UtilityIndices$ will be pruned prior to the first iteration $\TimeIndex$ with cumulative sample size $\CumulativeSamples_\TimeIndex \geq$
    \begin{equation*}
        2 + 2\ln\smash{\frac{3\abs{\UtilityIndices}\ScheduleLength}{\delta}}\min\begin{cases}
            \frac{10\UtilityRange}{\Regret_\PlayerIndex(\StratProfile; \Utility) - \gamma^*} + \frac{25\norm{\UtilityVariance_\PlayerIndex(\Adjacent_{\PlayerIndex, \StratProfile})}_\infty}{(\Regret_\PlayerIndex(\StratProfile;\Utility) - \gamma^*)^2}\\
            \frac{5\UtilityRange}{2\epsilon} + \frac{\UtilityVariance_\PlayerIndex(\StratProfile)}{\epsilon^{2}}
        \end{cases}
    \end{equation*}
    (defaulting to the second option when $\Regret_\PlayerIndex(\StratProfile; \Utility) \leq \gamma^*$). 

    When running \PSRegM{}$(\ConditionalGame{\ConditionSpace}, \ConditionDistribution, \UtilityIndices, \delta, \epsilon)$, with probability at least $1 - \frac{\delta}{3}$, the index $(\PlayerIndex, \StratProfile)\in \UtilityIndices$ will be pruned prior to the first iteration $\TimeIndex$ with cumulative sample size $\CumulativeSamples_\TimeIndex \geq$
    \begin{equation*}
        2 + 2\ln\smash{\frac{3\abs{\UtilityIndices}\ScheduleLength}{\delta}}\min\begin{cases}
            \frac{12.5\UtilityRange}{\Regret_\PlayerIndex(\StratProfile; \Utility) - \epsilon} + \frac{25\norm{\UtilityVariance_\PlayerIndex(\Adjacent_{\PlayerIndex, \StratProfile})}_\infty}{(\Regret_\PlayerIndex(\StratProfile;\Utility) - \epsilon)^2}\\
            \frac{5\UtilityRange}{2\epsilon} + \frac{\UtilityVariance_\PlayerIndex(\StratProfile)}{\epsilon^{2}}
        \end{cases}
    \end{equation*}
    queries at profile $\StratProfile$ (defaulting to the second option when $\Regret_\PlayerIndex(\StratProfile; \Utility) \leq \epsilon$). 
\end{restatable}

The above result reinforces the idea that care is required when choosing a sampling schedule for these algorithms. If the marginal sample size $\NumberOfSamples_\TimeIndex$ is small at each iteration $\TimeIndex$, then very few queries will be wasted from when an index is ready to be pruned to when it is actually pruned by the algorithm. On the other hand, if very small marginal sample sizes are used, then a very large schedule length $\ScheduleLength$ will be required to reach a sufficiently large cumulative sample size to prune all utility indices. But a larger schedule length $\ScheduleLength$ yields looser utility deviation bounds, thus resulting in \emph{all} indices requiring more queries to be pruned than otherwise. Hence, there is a trade-off in designing a sampling schedule between keeping marginal sample sizes small and keeping the schedule length small. 

We discuss our particular choice of sampling schedule in the next section, and in greater detail in the Appendix. In the following result, we use Hoeffding's Inequality to derive an upper bound on the requisite total cumulative sample size a sampling schedule needs to ensure that all utility indices will be pruned prior to its exhaustion.

\begin{restatable}{theorem}{omegaBound}
    \label{thm:omega}
    Suppose that the total samples $\CumulativeSamples_\ScheduleLength$ allocated in the sampling schedule is greater than or equal to the maximum number of samples needed to prune an arbitrary index, i.e., 
    \begin{equation*}
        \CumulativeSamples_\ScheduleLength\doteq \sum_{\TimeIndex=1}^\ScheduleLength \NumberOfSamples_\TimeIndex \geq \frac{\UtilityRange^2\ln\frac{2\abs{\UtilityIndices}\ScheduleLength}{\delta}}{2\epsilon^2}\enspace.
    \end{equation*}
    Then for each algorithm $\textup{PS}\in$
    $$\{\textup{\PSWE{}, \PSRegOld{}, \PSRegU{}, \PSRegNU{}, \PSRegM{}}\},$$
    it is guaranteed that PS$(\ConditionalGame{\ConditionSpace}, \ConditionDistribution, \UtilityIndices, \delta, \epsilon)$ will terminate and return an empirical game with the guarantees corresponding to the respective algorithm.
\end{restatable}

\section{Experiments}

\begin{figure}
    \centering
    \includegraphics[width=\columnwidth]{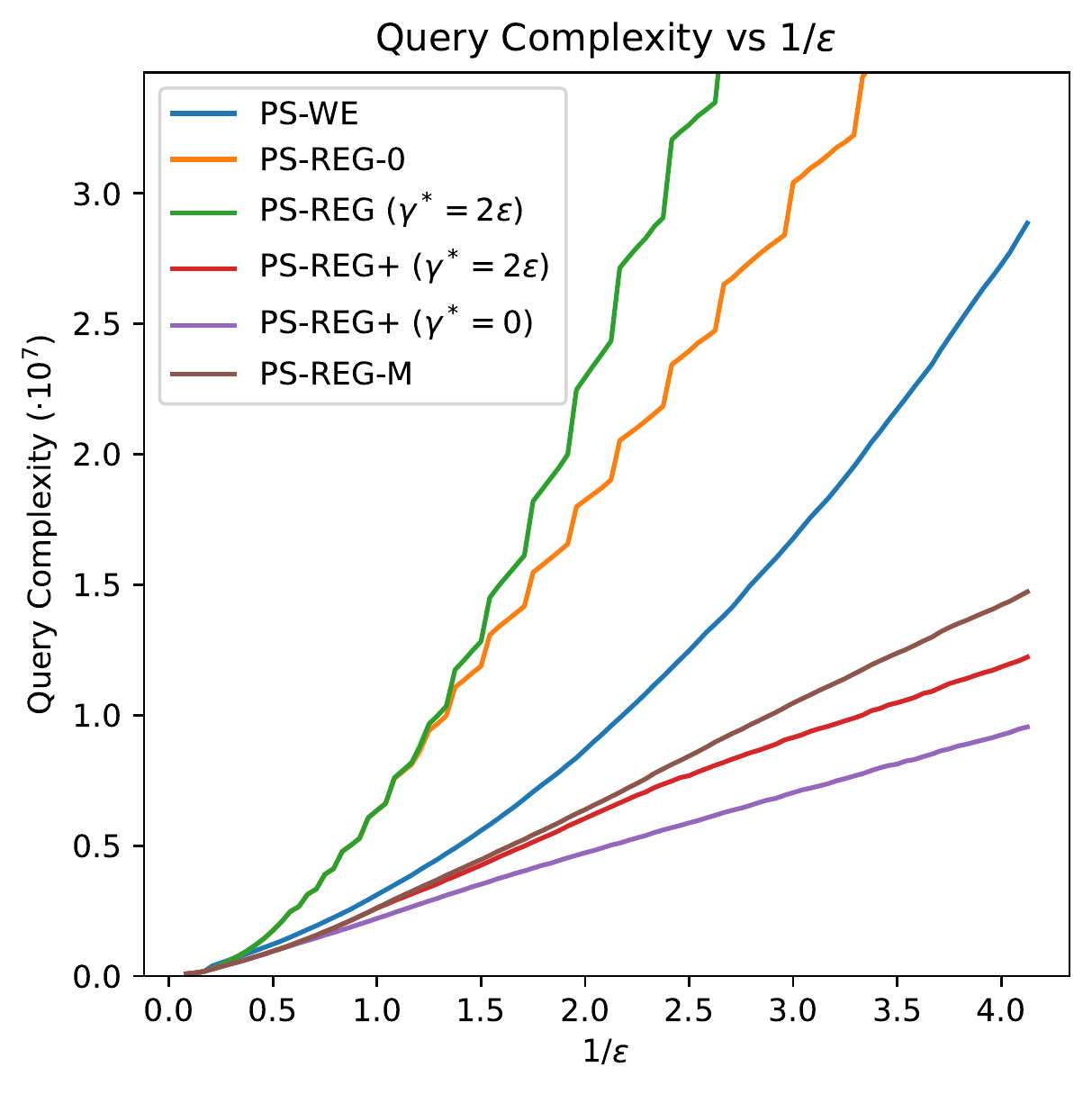}
    \caption{Average query complexity vs $\nicefrac{1}{\epsilon}$ for 10 runs of each algorithm for each target error $\epsilon\in \{\frac{\UtilityRange}{2}, \frac{\UtilityRange}{3}, \dots, \frac{\UtilityRange}{100}\}$. All algorithm runs are conducted on a single randomly-generated two-player zero-sum game with non-uniform additive noise.}
    \label{fig:query-one-game} 
\end{figure}

In this section, we experimentally explore the behavior of all the aforementioned PS algorithms. When choosing a sampling schedule for each algorithm, we follow \citet{cousins2022computational} in designing sampling schedules which begin at sample size $\alpha$, a lower bound on the minimum number of samples needed to prune an arbitrary index, and end with a cumulative sample size $\omega$, an upper bound on the maximum number of samples needed to prune an arbitrary index. 
For all algorithms, we use the upper bound from \Cref{thm:omega} for $\omega$, i.e., $\omega\doteq \frac{\UtilityRange^2\ln\frac{2\abs{\GameTuple}\ScheduleLength}{\delta}}{2\epsilon^2}$, thus guaranteeing that each algorithm will return an empirical game satisfying the respective guarantees of the algorithm upon termination. 
For algorithms using regret pruning, we set $\alpha$ to be a lower bound on the number of samples to estimate a zero-variance utility to (w.h.p.) within $\frac{\UtilityRange}{2}$ error (since no regret pruning can occur prior to at least one index achieving such an error guarantee; see Appendix for proof). For \PSWE{}, we follow \citet{cousins2022computational} in setting $\alpha$ to be a lower bound on the number of samples needed to estimate a zero-variance utility to (w.h.p.) within a target error $\epsilon$, though we improve their lower bound by a small constant factor. 
Finally, while our \PSWE{} sampling schedule has a geometric sampling schedule (i.e., geometrically increasing cumulative sample size) as in \citet{cousins2022computational}, using such a geometric schedule for our regret pruning algorithms results in too many iterations spent on very small sample sizes, yielding looser bounds with very few additional queries saved in return. To correct for this, our regret pruning algorithms use a sampling schedule which is linear until it reaches the corresponding $\alpha$ derived for \PSWE{}, and then follows the same geometric sampling schedule used in \PSWE{}. For more details regarding the sampling schedules, see the Appendix.

\subsection{Query Complexity vs. Target Error}

In the following experiments, we test our algorithms on two-player random zero-sum games (generated via the game-generator GAMUT \citep{nudelman2004run}) with 40 actions for each player and utility values in the range $[-2, 2]$. In order to emulate a noisy simulator, we add noise to each sample utility value. For each utility index $(\PlayerIndex, \StratProfile)\in\UtilityIndices$, we sample a variance modifier $\nu_{\PlayerIndex, \StratProfile}\sim \textup{Beta(1.5, 3)}$, and then each time the simulator is queried for $\SampleUtility(\StratProfile)\sim\Simulator(\StratProfile)$, we set $\SampleUtility_\PlayerIndex(\StratProfile) \doteq \Utility_\PlayerIndex(\StratProfile) + \mathcal{N}(\nu_{\PlayerIndex, \StratProfile})$ where $\mathcal{N}(\nu_{\PlayerIndex, \StratProfile})$ is a scaled and shifted Bernoulli random variable, generating either $10\nu_{\PlayerIndex, \StratProfile}$ or $-10\nu_{\PlayerIndex, \StratProfile}$ with equal probability. Notice then that our final utility range for these random zero-sum simulation-based games is $\UtilityRange = 24$. Sampling variance modifiers from $\textup{Beta}(1.5, 3)$ ensures that our utility indices have a wide range of noise variables with mostly moderate variance, but with some noise variables having particularly high variance and some particularly low variance.

For our first experiment (\Cref{fig:query-one-game}), we compare the query complexities (i.e., the number of simulation queries placed prior to termination) of our algorithms for varying target errors $\epsilon$. We begin by generating a two-player random zero-sum game $\GameTuple_{\Simulator}$ (and variance modifiers $\nu_{\PlayerIndex, \StratProfile}$ for each utility index $(\PlayerIndex, \StratProfile)\in\UtilityIndices$). For each target error $\epsilon\in \{\frac{\UtilityRange}{2}, \frac{\UtilityRange}{3}, \dots, \frac{\UtilityRange}{100}\}$, we then run each of our aforementioned PS algorithms (with failure probability $\delta\doteq 0.05$) on $\GameTuple_\Simulator$ a total of 10 times and plot the average query complexity across those runs. For each algorithm, we then connect these average query complexities for each target error by a line plot in \Cref{fig:query-one-game}.

In \Cref{fig:query-one-game}, we see that, as expected, \PSRegU{} obtains its guarantees for approximate pure equilibria of the empirical game at the cost of a slightly greater query complexity than \PSRegOld{}. In a similar vein, \PSRegNU{} with $\gamma^* = 2\epsilon$ also requires a greater number of queries than \PSRegNU{} with $\gamma^* = 0$ (i.e, \PSRegOld{}+) in order to yield its guarantees for $\gamma^*$-pure equilibria of the empirical game. We also observe that \PSRegOld{} and \PSRegU{} with $\gamma^* = 2\epsilon$, which both use uniform utility deviation bounds, consume significantly more queries than \PSWE{}, which takes advantage of non-uniform bounds but conducts no regret pruning. This suggests that utilizing non-uniform utility deviation bounds is crucial for designing query efficient progressive sampling algorithms. This idea is further reinforced when looking at \PSRegNU{} and \PSRegM{}, both of which use regret pruning criteria that take advantage of non-uniform utility deviation bounds, and outperform \PSWE{} by a very significant margin, especially when the target error $\epsilon$ is small. Another particularly surprising result is that \PSRegM{} only consumes marginally more queries than \PSRegNU{} with $\gamma^* = 2\epsilon$, despite yielding strong mixed Nash containment guarantees in return. 

\begin{figure}[!b]
    \centering
    \includegraphics[width=\columnwidth]{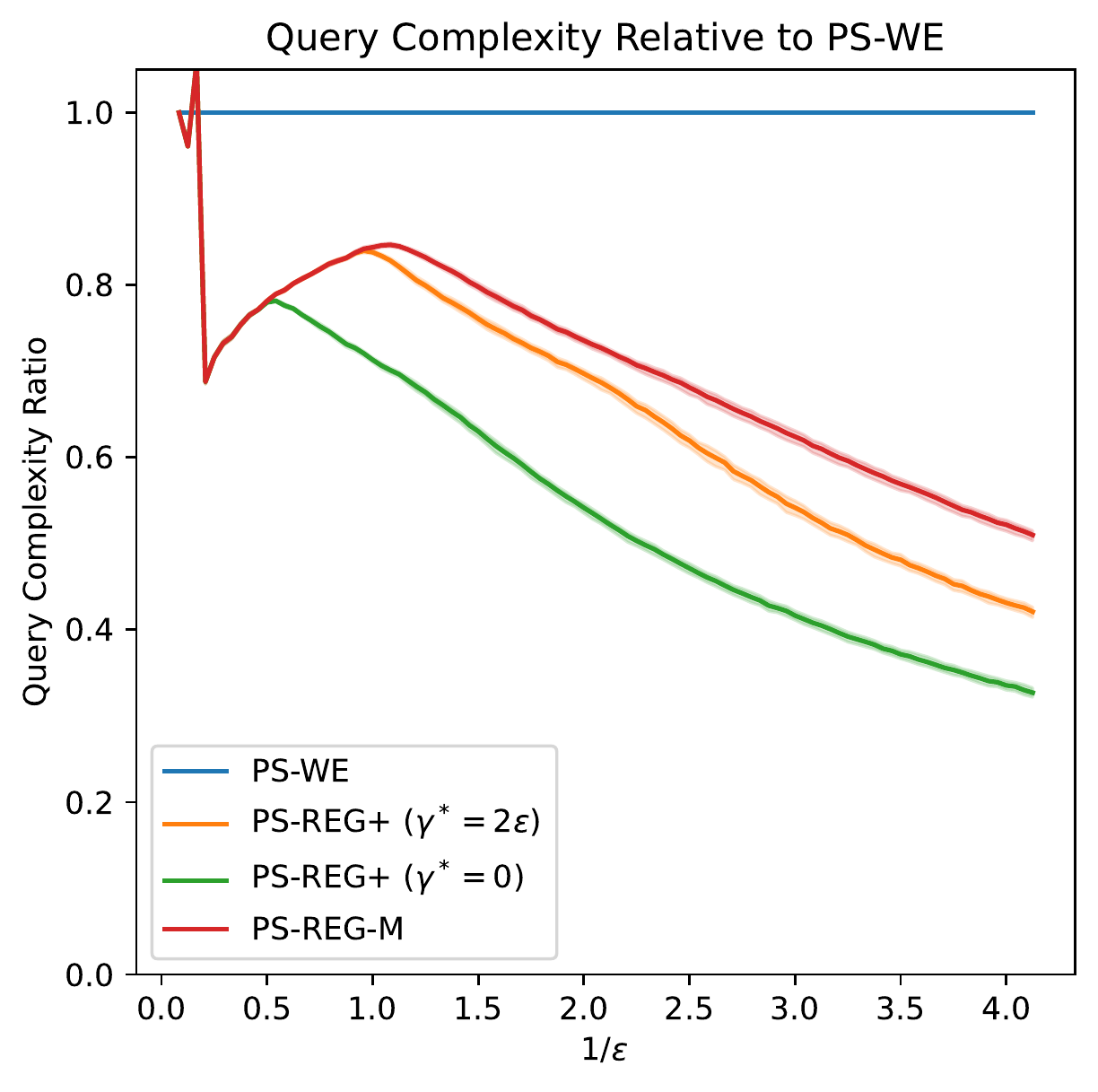}
    \caption{Average ratio of query complexity to query complexity of \PSWE{} vs $\nicefrac{1}{\epsilon}$ for a single run of each algorithm on each of 10 randomly generated two-player zero-sum games with non-uniform additive noise for each target error $\epsilon\in \{\frac{\UtilityRange}{2}, \frac{\UtilityRange}{3}, \dots, \frac{\UtilityRange}{100}\}$. Standard deviation bands are plotted and are just barely visible due to low variation.}
    \label{fig:query-many-games}
\end{figure}

Since our first experiment only tests on a single randomly generated simulation-based game, it is possible that the generated game was just particularly amenable to regret pruning. For our second experiment (\Cref{fig:query-many-games}), we observe the proportion of additional queries our new PS algorithms are able to save on average (across 10 random zero-sum simulation-based games) in comparison to \PSWE{}. This time, we run each PS algorithm (again with failure probability $\delta\doteq 0.05$) only once for each generated game.

In \Cref{fig:query-many-games}, we observe that the comparative performance of the algorithms in \Cref{fig:query-one-game} remain consistent across many different random two-player zero-sum games. We further observe that past a certain turning point, progressive sampling algorithms which use regret pruning techniques save a greater proportion of queries with respect to \PSWE{} as smaller target errors $\epsilon$ are used. When the target error $\epsilon$ is very small, \PSRegNU{} (both $\gamma^* = 0$ and $\gamma^* = 2\epsilon$) and \PSRegM{} are able to obtain their respective guarantees while saving more than 50\% and 40\%, respectively, of the queries used by \PSWE{}.

\section{Conclusion}

In this paper, we address a serious limitation of \citeauthor{areyan2020improved}'s progressive sampling algorithm with regret pruning -- it is only able to yield guarantees regarding \emph{true} pure Nash equilibria of the empirical game. We design two primary novel progressive sampling algorithms for practitioners to use to learn equilibria in simulation-based games: \PSRegNU{} and \PSRegM{}. \PSRegNU{} combines well-estimated pruning with a novel regret pruning variation which is modified to ensure the algorithm yields pure Nash containment guarantees for approximate $\gamma^*$-pure Nash equilibria of the empirical game and to take advantage of non-uniform utility deviation bounds to prune utility indices as soon as possible. When using \PSRegNU{}, a game analyst will set $\gamma^*$ according to the weakest approximate pure Nash equilibria for which they desire guarantees. \PSRegM{} also incorporates well-estimated pruning and a novel regret pruning variation, except unlike \PSRegNU{}, it yields strong Nash containment guarantees for all approximate pure or mixed Nash equilibria of the empirical game, at the cost of a slightly greater query complexity. Both \PSRegNU{} and \PSRegM{} significantly outperform \PSWE{}, the prior state-of-the-art algorithm for learning equilibria in simulation-based games. In light of this, game analysts seeking such an algorithm should use \PSRegNU{} if they only aim to learn pure Nash equilibria and should otherwise use \PSRegM{}.

In this work, we have only applied our progressive sampling algorithms and pruning techniques to normal-form games. In future work, we aim to extend this methodology to other game models such as extensive-form games. Additionally, EGTA algorithms for learning equilibria of simulation-based game have thus far been completely detached from algorithms for computing equilibria. A game analyst must first use EGTA algorithms to learn a sufficiently strong approximation of the simulation-based game, and then must compute equilibria of this empirical game. Future work can incorporate variance-sensitive and regret-sensitive progressive sampling techniques into an existing game-solving algorithm to make an EGTA-aware game-solving algorithm.

% Use \bibliography{yourbibfile} instead or the References section will not appear in your paper
\bibliography{aaai23}

\appendix
\onecolumn

\renewcommand\thefigure{\thesection.\arabic{figure}}
\renewcommand\thealgorithm{\thesection.\arabic{algorithm}}

\section{Appendix}
\label{sec:appendix}

\subsection{Notation}

To improve readability of the following proofs, we introduce a few notational short-hands.
\begin{enumerate}
    \item Given a utility index $(\PlayerIndex, \StratProfile)\in\UtilityIndices$ and a utility function $\Utility$, let $\StratProfile^*_{\Utility}$ denote the best response of player $\PlayerIndex$ to the opponents' strategies in $\StratProfile$, i.e., $\StratProfile^*_{\Utility} \doteq \sup_{\StratProfile' \in \Adjacent_{\PlayerIndex, \StratProfile}} \Utility_{\PlayerIndex} (\StratProfile')$.
    \item Given a mixed strategy profile $\StratProfile\in\MixedStratProfileSpace$, player $\PlayerIndex\in\SetOfPlayers$, and strategy $\StrategyAlt\in\StrategySet_\PlayerIndex$, we let $\StratProfile\vert_\StrategyAlt$ denote the strategy profile $\StratProfile'\in\Adjacent_{\PlayerIndex, \StratProfile}$ satisfying $\StratProfile'_\PlayerIndex = \StrategyAlt$, and let $\Prob[\StrategyAlt \vert \StratProfile]$ denote the probability that mixed strategy profile $\StratProfile$ assigns to strategy $\StrategyAlt$.
\end{enumerate}

\subsection{Nash Containment Lemmas}

We begin by proving the three different Nash containment lemmas.

\tuylsDual*
\begin{proof}
    Suppose that $\abs{\Utility_\PlayerIndex(\StratProfile) - \EUtility_\PlayerIndex(\StratProfile)}\leq \epsilon$ for all $(\PlayerIndex, \StratProfile)\in\UtilityIndices$.
    Let $\gamma\geq 0$ be arbitrary, and suppose that $\StratProfile\in\Equilibria_\gamma(\Utility)$. Then we have that
    \begin{align*}
        \Regret_\PlayerIndex(\StratProfile; \EUtility) &= \EUtility_\PlayerIndex(\StratProfile^*_{\EUtility}) - \EUtility_\PlayerIndex(\StratProfile)\\
        &\leq (\Utility_\PlayerIndex(\StratProfile^*_{\EUtility}) + \epsilon) - (\Utility_\PlayerIndex(\StratProfile) - \epsilon)& &\textsf{By assumption}\\
        &\leq (\Utility_\PlayerIndex(\StratProfile) + \gamma + \epsilon) - (\Utility_\PlayerIndex(\StratProfile) - \epsilon)&
        &\textsf{Since $\StratProfile\in\Equilibria_\gamma(\Utility)$, it holds that $\Utility_\PlayerIndex(\StratProfile) + \gamma \geq \Utility_\PlayerIndex(\StratProfile^*_{\EUtility})$}\\
        &= 2\epsilon + \gamma,
    \end{align*}
    and hence $\StratProfile\in\Equilibria_{2\epsilon + \gamma}(\EUtility)$. Since $\StratProfile$ was arbitrary, we have that $\Equilibria_\gamma(\Utility)\subseteq \Equilibria_{2\epsilon + \gamma}(\EUtility)$.
    By completely analogous reasoning, we see that $\Equilibria_\gamma(\EUtility)\subseteq \Equilibria_{2\epsilon + \gamma}(\Utility)$.
\end{proof}

\dualNew*
\begin{proof}
    Suppose that $\abs{\Utility_\PlayerIndex(\StratProfile) - \EUtility_\PlayerIndex(\StratProfile)}\leq \epsilon$ for all $(\PlayerIndex, \StratProfile)\in\UtilityIndices$ satisfying $\Regret_\PlayerIndex(\StratProfile; \Utility) = 0$ or $\Regret_\PlayerIndex(\StratProfile; \EUtility) \leq \gamma^*$. Suppose that $\StratProfile\in\Equilibria(\Utility)$. Then we have that
    \begin{align*}
        \Regret_\PlayerIndex(\StratProfile; \EUtility) &= \EUtility_\PlayerIndex(\StratProfile^*_{\EUtility}) - \EUtility_\PlayerIndex(\StratProfile)\\
        &\leq (\Utility_\PlayerIndex(\StratProfile^*_{\EUtility}) + \epsilon) - (\Utility_\PlayerIndex(\StratProfile) - \epsilon)&
        &\begin{aligned}[t]
            &\textsf{$\Regret_\PlayerIndex(\StratProfile^*_{\EUtility}; \EUtility) = 0 \leq \gamma^*$, so $(\PlayerIndex, \StratProfile^*_{\EUtility})$ is well-estimated (WE)}\\
            &\textsf{$\Regret_\PlayerIndex(\StratProfile; \Utility) = 0$ (since $\StratProfile\in\Equilibria(\Utility)$), so $(\PlayerIndex, \StratProfile)$ is WE}
        \end{aligned}\\
        &\leq (\Utility_\PlayerIndex(\StratProfile) + \epsilon) - (\Utility_\PlayerIndex(\StratProfile) - \epsilon)&
        &\textsf{Since $\StratProfile\in\Equilibria(\Utility)$, it holds that $\Utility_\PlayerIndex(\StratProfile) \geq \Utility_\PlayerIndex(\StratProfile^*_{\EUtility})$}\\
        &= 2\epsilon,
    \end{align*}
    and hence $\StratProfile\in\Equilibria_{2\epsilon}(\EUtility)$. Since $\StratProfile$ was arbitrary, we have that $\Equilibria(\Utility)\subseteq \Equilibria_{2\epsilon}(\EUtility)$.
    Now let $\gamma\in[0, \gamma^*]$ and suppose instead that $\StratProfile\in\Equilibria_{\gamma}(\EUtility)$. Then we have that
    \begin{align*}
        \Regret_\PlayerIndex(\StratProfile; \Utility) &= \Utility_\PlayerIndex(\StratProfile^*_{\Utility}) - \Utility_\PlayerIndex(\StratProfile)\\
        &\leq (\EUtility_\PlayerIndex(\StratProfile^*_{\Utility}) + \epsilon) - (\EUtility_\PlayerIndex(\StratProfile) - \epsilon)&
        &\begin{aligned}[t]
            &\textsf{$\Regret_\PlayerIndex(\StratProfile^*_{\Utility}; \Utility) = 0$ (by definition), so $(\PlayerIndex, \StratProfile^*_{\Utility})$ is WE}\\
            &\textsf{$\Regret_\PlayerIndex(\StratProfile; \EUtility) \leq \gamma^*$ (since $\StratProfile\in\Equilibria_{\gamma}(\EUtility)\subseteq \Equilibria_{\gamma^*}(\EUtility)$), so $(\PlayerIndex, \StratProfile)$ is WE}
        \end{aligned}\\
        &\leq (\EUtility_\PlayerIndex(\StratProfile) + \gamma + \epsilon) - (\EUtility_\PlayerIndex(\StratProfile) - \epsilon)&
        &\textsf{Since $\StratProfile\in\Equilibria_{\gamma}(\EUtility)$, it holds that $\EUtility_\PlayerIndex(\StratProfile) + \gamma \geq \EUtility_\PlayerIndex(\StratProfile^*_{\Utility})$}\\
        &= 2\epsilon + \gamma,
    \end{align*}
    and hence $\StratProfile\in\Equilibria_{2\epsilon + \gamma}(\Utility)$. Since $\StratProfile$ was arbitrary, we have that $\Equilibria_{\gamma}(\EUtility)\subseteq \Equilibria_{2\epsilon + \gamma}(\Utility)$ and we are done.
\end{proof}

Before proving the third Nash containment lemma, we prove an accessory lemma.

\begin{lemma}
    \label{lem:zero-regret}
    Let $\epsilon > 0$ be arbitrary. Suppose that for all $(\PlayerIndex, \StratProfile)\in\UtilityIndices$, it holds that
    \begin{equation*}
        \abs{\Utility_p(\StratProfile) - \EUtility_p(\StratProfile)} \leq \max\left\{\epsilon, \frac{\Regret_\PlayerIndex(\StratProfile; \EUtility)}{2}\right\}\enspace.
    \end{equation*}
    Then for all $(\PlayerIndex, \StratProfile)\in\UtilityIndices$ satisfying $\Regret_\PlayerIndex(\StratProfile; \Utility) = 0$, it must hold that $\abs{\Utility_p(\StratProfile) - \EUtility_p(\StratProfile)}\leq \epsilon$.
\end{lemma}
\begin{proof}
    Suppose that $(\PlayerIndex, \StratProfile)\in\UtilityIndices$ satisfies $\Regret_\PlayerIndex(\StratProfile; \Utility) = 0$. Since $\Regret_\PlayerIndex(\StratProfile^*_{\EUtility}; \EUtility) = 0$, we have that $\abs{\Utility_\PlayerIndex(\StratProfile^*_{\EUtility}) - \EUtility_\PlayerIndex(\StratProfile^*_{\EUtility})} \leq \epsilon$, which implies that $\Utility_\PlayerIndex(\StratProfile^*_{\EUtility}) \geq \EUtility_\PlayerIndex(\StratProfile^*_{\EUtility}) - \epsilon$. But by definition, we have that $\Utility_\PlayerIndex(\StratProfile) \geq \Utility_\PlayerIndex(\StratProfile^*_{\EUtility})$, and hence, it must hold that $\Utility_\PlayerIndex(\StratProfile)\geq \EUtility_\PlayerIndex(\StratProfile^*_{\EUtility}) - \epsilon$. By hypothesis, we also have that $\Utility_\PlayerIndex(\StratProfile) \leq \EUtility_\PlayerIndex(\StratProfile) + \max\left\{\epsilon, \frac{\Regret_\PlayerIndex(\StratProfile; \EUtility)}{2}\right\}$. Chaining these two inequalities, we have that
    \begin{alignat*}{3}
        & & \EUtility_\PlayerIndex(\StratProfile^*_{\EUtility}) - \epsilon &\leq \EUtility_\PlayerIndex(\StratProfile) + \max\left\{\epsilon, \frac{\Regret_\PlayerIndex(\StratProfile; \EUtility)}{2}\right\}\\
        &\Longleftrightarrow\quad& \EUtility_\PlayerIndex(\StratProfile^*_{\EUtility}) - \EUtility_\PlayerIndex(\StratProfile) &\leq \epsilon + \max\left\{\epsilon, \frac{\Regret_\PlayerIndex(\StratProfile; \EUtility)}{2}\right\}\\
        &\Longleftrightarrow\quad& \Regret_\PlayerIndex(\StratProfile; \EUtility) &\leq \epsilon + \max\left\{\epsilon, \frac{\Regret_\PlayerIndex(\StratProfile; \EUtility)}{2}\right\}\enspace.\\
    \end{alignat*}
    If $\frac{\Regret_\PlayerIndex(\StratProfile; \EUtility)}{2} > \epsilon$, then the second term in the max wins out, but solving for $\Regret_\PlayerIndex(\StratProfile; \EUtility)$ immediately yields that $\Regret_\PlayerIndex(\StratProfile; \EUtility) \leq 2\epsilon$, a contradiction. Hence, it must hold that $\frac{\Regret_\PlayerIndex(\StratProfile; \EUtility)}{2}\leq \epsilon$, and hence that $\abs{\Utility_p(\StratProfile) - \EUtility_p(\StratProfile)}\leq \epsilon$.
\end{proof}

\dualMiddle* 
\begin{proof}
    Suppose that $\abs{\Utility_\PlayerIndex(\StratProfile) - \EUtility_\PlayerIndex(\StratProfile)}\leq \max\left\{\epsilon, \frac{\Regret_\PlayerIndex(\StratProfile; \EUtility)}{2}\right\}$ for all $(\PlayerIndex, \StratProfile)\in\UtilityIndices$.
    We will first show the pure Nash containment result, and then show the mixed Nash containment result.
    Suppose that $\StratProfile\in\Equilibria({\Utility})$. Then we have that
    \begin{align*}
        \Regret_\PlayerIndex(\StratProfile; \EUtility) &= \EUtility_\PlayerIndex(\StratProfile^*_{\EUtility}) - \EUtility_\PlayerIndex(\StratProfile)&
        &\textsf{By definition}\\
        &\leq (\Utility_\PlayerIndex(\StratProfile^*_{\EUtility}) + \epsilon) - (\Utility_\PlayerIndex(\StratProfile) - \epsilon)&
        &\begin{aligned}[t]
            &\textsf{Since $\Regret_\PlayerIndex(\StratProfile^*_{\EUtility}; \EUtility) = 0$, by hypothesis $(\PlayerIndex, \StratProfile^*_{\EUtility})$ is WE}\\
            &\textsf{Since $\Regret_\PlayerIndex(\StratProfile; \Utility) = 0$, by \Cref{lem:zero-regret} $(\PlayerIndex, \StratProfile)$ is WE.}
        \end{aligned}\\
        &\leq (\Utility_\PlayerIndex(\StratProfile) + \epsilon) - (\Utility_\PlayerIndex(\StratProfile) - \epsilon)&
        &\textsf{Since $\StratProfile\in\Equilibria(\Utility)$, it holds that $\Utility_\PlayerIndex(\StratProfile) \geq \Utility_\PlayerIndex(\StratProfile^*_{\EUtility})$}\\
        &= 2\epsilon,
    \end{align*}
    and hence $\StratProfile\in \Equilibria_{2\epsilon}(\EUtility)$. Since $\StratProfile$ was arbitrary, we have that $\Equilibria(\Utility)\subseteq \Equilibria_{2\epsilon}(\EUtility)$.

    Now let $\gamma\in[0, 2\epsilon]$ suppose instead that $\StratProfile\in\Equilibria_{\gamma}(\EUtility)$. Then we have that
    \begin{align*}
        \Regret_\PlayerIndex(\StratProfile; \Utility) &= \Utility_\PlayerIndex(\StratProfile^*_{\Utility}) - \Utility_\PlayerIndex(\StratProfile)& &\textsf{By definition}\\
        &\leq (\EUtility_\PlayerIndex(\StratProfile^*_{\Utility}) + \epsilon) - (\EUtility_\PlayerIndex(\StratProfile) - \epsilon)&
        &\begin{aligned}[t]
            &\textsf{Since $\Regret_\PlayerIndex(\StratProfile^*_{\Utility}; \Utility) = 0$, by \Cref{lem:zero-regret} $(\PlayerIndex, \StratProfile^*_{\Utility})$ is WE}\\
            &\textsf{Since $\Regret_\PlayerIndex(\StratProfile; 
\EUtility) \leq \gamma\leq 2\epsilon$ (by def. of $\Equilibria_{\gamma}(\EUtility)$), by hypothesis $(\PlayerIndex, \StratProfile)$ is WE}
        \end{aligned}\\
        &\leq (\EUtility_\PlayerIndex(\StratProfile) + \gamma + \epsilon) - (\EUtility_\PlayerIndex(\StratProfile) - \epsilon)&
        &\textsf{Since $\StratProfile\in\Equilibria_{\gamma}(\EUtility)$, it holds that $\EUtility_\PlayerIndex(\StratProfile) + \gamma \geq \EUtility_\PlayerIndex(\StratProfile^*_{\Utility})$}\\
        &= 2\epsilon + \gamma,
    \end{align*}
    and hence $\StratProfile\in\Equilibria_{2\epsilon + \gamma}(\Utility)$. Since $\StratProfile$ was arbitrary, we have that $\Equilibria_{\gamma}(\EUtility)\subseteq \Equilibria_{2\epsilon + \gamma}(\Utility)$, and we have shown the pure Nash containment result.

    Now we will prove the mixed Nash containment result. Suppose that $\StratProfile\in\MixedEquilibria(\Utility)$. Then we have that
    \begin{align*}
        \Regret_\PlayerIndex(\StratProfile; \EUtility) &= \EUtility_\PlayerIndex(\StratProfile^*_{\EUtility}) - \EUtility_\PlayerIndex(\StratProfile)& &\textsf{By definition}\\
        &= \EUtility_\PlayerIndex(\StratProfile^*_{\EUtility}) - \sum_{\StrategyAlt\in\StrategySet_\PlayerIndex}\EUtility_\PlayerIndex(\StratProfile\vert_\StrategyAlt)\Prob[\StrategyAlt \vert \StratProfile]& &\textsf{By definition (of utility at mixed strategy profile)}\\
        &\leq (\Utility_\PlayerIndex(\StratProfile^*_{\EUtility}) + \epsilon) - \sum_{\StrategyAlt\in\StrategySet_\PlayerIndex}\left(\Utility_\PlayerIndex(\StratProfile\vert_\StrategyAlt) - \max\left\{\epsilon, \frac{\Regret_\PlayerIndex(\StratProfile\vert_\StrategyAlt; \EUtility)}{2}\right\}\right)\Prob[\StrategyAlt \vert \StratProfile]& & \textsf{By hypothesis}\\
        &\leq (\Utility_\PlayerIndex(\StratProfile) + \epsilon) - \sum_{\StrategyAlt\in\StrategySet_\PlayerIndex}\left(\Utility_\PlayerIndex(\StratProfile\vert_\StrategyAlt) - \max\left\{\epsilon, \frac{\Regret_\PlayerIndex(\StratProfile\vert_\StrategyAlt; \EUtility)}{2}\right\}\right)\Prob[\StrategyAlt \vert \StratProfile]& &\textsf{Since $\StratProfile\in\MixedEquilibria(\Utility)$, we have $\Utility_\PlayerIndex(\StratProfile)\geq \Utility_\PlayerIndex(\StratProfile^*_{\EUtility})$}\\
        &= \epsilon + \sum_{\StrategyAlt\in\StrategySet_\PlayerIndex}\Prob[\StrategyAlt \vert \StratProfile]\max\left\{\epsilon, \frac{\Regret_\PlayerIndex(\StratProfile\vert_\StrategyAlt; \EUtility)}{2}\right\}& &\textsf{By definition (of utility at mixed strategy profile)}\\
        &\leq \epsilon + \sum_{\StrategyAlt\in\StrategySet_\PlayerIndex}\Prob[\StrategyAlt \vert \StratProfile]\left(\epsilon + \frac{\Regret_\PlayerIndex(\StratProfile\vert_\StrategyAlt; \EUtility)}{2}\right)& &\textsf{$\max\{A, B\}\leq A + B$}\\
        &= 2\epsilon + \frac{\Regret_\PlayerIndex(\StratProfile; \EUtility)}{2}.& &\textsf{By definition (of utility at mixed strategy profile)}
    \end{align*}
    Solving the above for $\Regret_\PlayerIndex(\StratProfile; \EUtility)$, we get that $\Regret_\PlayerIndex(\StratProfile; \EUtility)\leq 4\epsilon$, and hence $\StratProfile\in\MixedEquilibria_{4\epsilon}(\EUtility)$. Since $\StratProfile$ was arbitrary, we have that $\MixedEquilibria(\Utility)\subseteq \MixedEquilibria_{4\epsilon}(\EUtility)$.

    Now instead let $\gamma \geq 0$, and suppose that $\StratProfile\in \MixedEquilibria_{\gamma}(\EUtility)$. Then we have that
    \begin{align*}
        \Regret_\PlayerIndex(\StratProfile; \Utility) &= \Utility_\PlayerIndex(\StratProfile^*_{\Utility}) - \Utility_\PlayerIndex(\StratProfile)& &\textsf{By definition}\\
        &= \Utility_\PlayerIndex(\StratProfile^*_{\Utility}) - \sum_{\StrategyAlt\in\StrategySet_\PlayerIndex}\Utility_\PlayerIndex(\StratProfile\vert_\StrategyAlt)\Prob[\StrategyAlt \vert \StratProfile]& &\textsf{By definition (of utility at mixed strategy profile)}\\
        &\leq (\EUtility_\PlayerIndex(\StratProfile^*_{\Utility}) + \epsilon) - \sum_{\StrategyAlt\in\StrategySet_\PlayerIndex}\left(\EUtility_\PlayerIndex(\StratProfile\vert_\StrategyAlt) - \max\left\{\epsilon, \frac{\Regret_\PlayerIndex(\StratProfile\vert_\StrategyAlt; \EUtility)}{2}\right\}\right)\Prob[\StrategyAlt \vert \StratProfile]& &\textsf{Hypothesis + \Cref{lem:zero-regret}}\\
        &\leq (\EUtility_\PlayerIndex(\StratProfile) + \gamma + \epsilon) - \sum_{\StrategyAlt\in\StrategySet_\PlayerIndex}\left(\EUtility_\PlayerIndex(\StratProfile\vert_\StrategyAlt) - \max\left\{\epsilon, \frac{\Regret_\PlayerIndex(\StratProfile\vert_\StrategyAlt; \EUtility)}{2}\right\}\right)\Prob[\StrategyAlt \vert \StratProfile]& &\textsf{Since $\StratProfile\in\MixedEquilibria_{\gamma}(\EUtility)$, we have $\EUtility_\PlayerIndex(\StratProfile) + \gamma \geq \EUtility_\PlayerIndex(\StratProfile^*_{\Utility})$}\\
        &= \gamma + \epsilon + \sum_{\StrategyAlt\in\StrategySet_\PlayerIndex}\Prob[\StrategyAlt \vert \StratProfile]\max\left\{\epsilon, \frac{\Regret_\PlayerIndex(\StratProfile\vert_\StrategyAlt; \EUtility)}{2}\right\}& &\textsf{By definition (of utility at mixed strategy profile)}\\
        &\leq \gamma + \epsilon + \sum_{\StrategyAlt\in\StrategySet_\PlayerIndex}\Prob[\StrategyAlt \vert \StratProfile]\left(\epsilon + \frac{\Regret_\PlayerIndex(\StratProfile\vert_\StrategyAlt; \EUtility)}{2}\right)& &\textsf{$\max\{A, B\}\leq A + B$}\\
        &= \gamma + 2\epsilon + \frac{\Regret_\PlayerIndex(\StratProfile; \EUtility)}{2}& &\textsf{By definition (of utility at mixed strategy profile)}\\
        &\leq \frac{3\gamma}{2} + 2\epsilon\enspace,& &\textsf{Since $\StratProfile\in\Equilibria_{\gamma}(\EUtility)$, we have $\Regret_\PlayerIndex(\StratProfile; \EUtility) \leq \gamma$}
    \end{align*}
    and hence $\StratProfile\in\MixedEquilibria_{2\epsilon + \frac{3\gamma}{2}}(\Utility)$. Since $\StratProfile$ was arbitrary, we have that $\MixedEquilibria_{\gamma}(\EUtility)\subseteq \MixedEquilibria_{2\epsilon + \frac{3\gamma}{2}}(\Utility)$.
\end{proof}

In the text, we claim that the condition in \Cref{lem:middle-dual-containment} is strictly tighter than the condition in \Cref{lem:new-dual-containment} when $\gamma^* = 2\epsilon$.
We show this by proving the forward implication from \Cref{lem:middle-dual-containment}'s condition to \Cref{lem:new-dual-containment}'s condition, and then showing the reverse implication to be false.

Suppose the condition from \Cref{lem:middle-dual-containment}, i.e., that $\abs{\Utility_\PlayerIndex(\StratProfile) - \EUtility_\PlayerIndex(\StratProfile)}\leq \max\left\{\epsilon, \frac{\Regret_\PlayerIndex(\StratProfile; \EUtility)}{2}\right\}$ for all $(\PlayerIndex, \StratProfile)\in\UtilityIndices$. Then by \Cref{lem:zero-regret}, we have that $\abs{\Utility_\PlayerIndex(\StratProfile) - \EUtility_\PlayerIndex(\StratProfile)}\leq \epsilon$ for all $(\PlayerIndex, \StratProfile)\in\UtilityIndices$ with $\Regret_\PlayerIndex(\StratProfile; \Utility) = 0$. But by the hypothesis, we directly have that for all $(\PlayerIndex, \StratProfile)\in\UtilityIndices$ with $\Regret_\PlayerIndex(\StratProfile; \EUtility)\leq 2\epsilon$, it holds that $\abs{\Utility_\PlayerIndex(\StratProfile) - \EUtility_\PlayerIndex(\StratProfile)} \leq \max\left\{\epsilon, \frac{\Regret_\PlayerIndex(\StratProfile; \EUtility)}{2}\right\} \leq \max\{\epsilon, \epsilon\} = \epsilon$. Hence, the condition from \Cref{lem:new-dual-containment} is satisfied.

But obviously the converse does not hold, since indices with $\Regret_\PlayerIndex(\StratProfile; \Utility) > 0$ and $\Regret_\PlayerIndex(\StratProfile; \EUtility) > 2\epsilon$ can have arbitrarily bad approximations (e.g., the empirical utilities can be arbitrarily low), even if the condition in \Cref{lem:new-dual-containment} holds.

\subsection{Algorithm Correctness Proofs}

A crucial component to all of our progressive sampling algorithm correctness proofs is the use of a union bound to ensure that (w.h.p.) all pruning that occurs is justified with respect to the true game. We see this line of reasoning in the following correctness proof for \PSWE{}.

\algoPSWE*
\begin{proof}
    Recall that for each index $(\PlayerIndex, \StratProfile)\in\UtilityIndices$ and iteration $\TimeIndex\in\{1, \dots, \ScheduleLength\}$, the scalar  $\hat{\epsilon}^{(\TimeIndex)}_\PlayerIndex(\StratProfile)$ is an upper bound on $\abs{\Utility_\PlayerIndex(\StratProfile) - \EUtility^{(\TimeIndex)}_\PlayerIndex(\StratProfile)}$ with probability at least $1 - \frac{\delta}{\abs{\GameTuple}\ScheduleLength}$. Thus, via a union bound, with probability $1 - \delta$, it holds \emph{for all} indices $(\PlayerIndex, \StratProfile)\in\UtilityIndices$ and iterations $\TimeIndex\in\{1, \dots, \ScheduleLength\}$ that if $\EUtility_\PlayerIndex^{(\TimeIndex)}(\StratProfile)$ has been computed (i.e., the index hasn't been pruned on a prior iteration), it satisfies $\abs{\Utility_\PlayerIndex(\StratProfile) - \EUtility^{(\TimeIndex)}_\PlayerIndex(\StratProfile)} \leq \hat{\epsilon}^{(\TimeIndex)}_\PlayerIndex(\StratProfile)$.

    Now suppose that \PSWE{} returns an empirical utility function $\EUtility$. This implies that \PSWE{} managed to prune all utility indices prior to the exhaustion of its sampling schedule. Combining this with the result from the previous paragraph, we get that with probability at least $1 - \delta$, for all $(\PlayerIndex, \StratProfile)\in\UtilityIndices$, if we let $\TimeIndex$ denote the iteration on which $(\PlayerIndex, \StratProfile)$ was well-estimated pruned, it holds that
    \begin{equation*}
        \abs{\Utility_\PlayerIndex(\StratProfile) - \EUtility_\PlayerIndex(\StratProfile)} = \abs{\Utility_\PlayerIndex(\StratProfile) - \EUtility^{(\TimeIndex)}_\PlayerIndex(\StratProfile)} \leq \hat{\epsilon}^{(\TimeIndex)}_\PlayerIndex(\StratProfile) \leq \epsilon.
    \end{equation*}
    The correctness guarantee then follows from \Cref{thm:dual-containment-tuyls}
\end{proof}

As the same union bound argument is applied in all of the following correctness proofs, we do not repeat it again. We begin by proving the correctness of our \PSRegNU{} algorithm, and show \PSRegU{} and \PSRegOld{} to simply be special cases of this algorithm. Before proving the correctness of \PSRegNU{}, we prove another accessory lemma.

\begin{lemma}
    \label{lem:sup-bounds}
    Suppose that for all $(\PlayerIndex, \StratProfile)\in\UtilityIndices$, it holds that $\abs{\Utility_\PlayerIndex(\StratProfile) - \EUtility_\PlayerIndex(\StratProfile)}\leq \epsilon_\PlayerIndex(\StratProfile)$. Then for all $(\PlayerIndex, \StratProfile)\in\UtilityIndices$, we have that
    \begin{equation*}
        \sup_{\StratProfile'\in \Adjacent_{\PlayerIndex, \StratProfile}} \left(\EUtility_\PlayerIndex(\StratProfile') - \epsilon_\PlayerIndex(\StratProfile')\right)\leq \Utility_\PlayerIndex(\StratProfile^*_{\Utility})\leq \sup_{\StratProfile'\in \Adjacent_{\PlayerIndex, \StratProfile}} \left(\EUtility_\PlayerIndex(\StratProfile') + \epsilon_\PlayerIndex(\StratProfile')\right)
    \end{equation*}
\end{lemma}
\begin{proof}
    Let $(\PlayerIndex, \StratProfile)\in\UtilityIndices$. We have that
    \begin{align*}
        \Utility_\PlayerIndex(\StratProfile^*_{\Utility}) = \sup_{\StratProfile'\in\Adjacent_{\PlayerIndex, \StratProfile}} \Utility_\PlayerIndex(\StratProfile') \geq \sup_{\StratProfile'\in\Adjacent_{\PlayerIndex, \StratProfile}} (\EUtility_\PlayerIndex(\StratProfile') - \epsilon_\PlayerIndex(\StratProfile')).
    \end{align*}
    The second inequality holds by analogous reasoning.
\end{proof}

\regretNU*
\begin{proof}
    Suppose that \PSRegNU{} returns an empirical utility function $\EUtility$, thus guaranteeing that all indices have been pruned.
    Suppose that an index $(\PlayerIndex, \StratProfile)\in\UtilityIndices$ is regret-pruned on iteration $\TimeIndex$. Then, we have (w.h.p.) that
    \begin{align*}
        \Regret_\PlayerIndex(\StratProfile, \Utility) &= \Utility_\PlayerIndex(\StratProfile^*_{\Utility}) - \Utility_\PlayerIndex(\StratProfile)\\
        &\geq \sup_{\StratProfile'\in\Adjacent_{\PlayerIndex, \StratProfile}}\left(\EUtility^{(\TimeIndex)}_\PlayerIndex(\StratProfile') - \hat{\epsilon}^{(\TimeIndex)}_\PlayerIndex(\StratProfile')\right) - \left(\EUtility^{(\TimeIndex)}_\PlayerIndex(\StratProfile) + \hat{\epsilon}^{(\TimeIndex)}_\PlayerIndex(\StratProfile)\right)& &\textsf{\Cref{lem:sup-bounds} + Definition of $\epsilon^{(\TimeIndex)}_\PlayerIndex(\StratProfile)$} \\
        &= \Regret^\downarrow_\PlayerIndex(\StratProfile, \EUtility^{(\TimeIndex)}) > 0& &\textsf{Definition of $\Regret^\downarrow$ + Pruning Criterion}.
    \end{align*}
    We also have (w.h.p.) that
    \begin{align*}
        \Regret_\PlayerIndex(\StratProfile, \EmpiricalGame{\Simulator}) &= \EUtility_\PlayerIndex(\StratProfile^*_{\EUtility}) - \EUtility_\PlayerIndex(\StratProfile)\\
        &=\EUtility_\PlayerIndex(\StratProfile^*_{\EUtility}) - \EUtility^{(\TimeIndex)}_\PlayerIndex(\StratProfile)& &\textsf{Since $(\PlayerIndex, \StratProfile)$ was pruned on iteration $i$}\\
        &\geq \EUtility_\PlayerIndex(\StratProfile^*_{\Utility}) - \EUtility^{(\TimeIndex)}_\PlayerIndex(\StratProfile)& &\textsf{$\EUtility_\PlayerIndex(\StratProfile^*_{\EUtility})\geq \EUtility_\PlayerIndex(\StratProfile^*_{\Utility})$ by definition}\\
        &\geq
        (\Utility_\PlayerIndex(\StratProfile^*_{\Utility}) - \epsilon) - \EUtility^{(\TimeIndex)}_\PlayerIndex(\StratProfile)&
        &\begin{aligned}[t]
            &\textsf{$(\PlayerIndex, \StratProfile^*_{\Utility})$ cannot have been regret pruned since $\Regret_\PlayerIndex(\StratProfile^*_{\Utility}, \GameTuple_{\Simulator}) = 0$.}\\
            &\textsf{Hence, it must have been WE pruned.}
        \end{aligned}\\
        &\geq \left(\sup_{\StratProfile'\in\Adjacent_{\PlayerIndex, \StratProfile}}\left(\EUtility^{(\TimeIndex)}_\PlayerIndex(\StratProfile') - \hat{\epsilon}_\PlayerIndex^{(\TimeIndex)}(\StratProfile')\right) - \epsilon\right) - \EUtility^{(\TimeIndex)}_\PlayerIndex(\StratProfile)& &\textsf{\Cref{lem:sup-bounds}}\\
        &= \Regret^\downarrow_\PlayerIndex(\StratProfile, \EmpiricalGame{\Simulator}^{(\TimeIndex)}) + \hat{\epsilon}^{(\TimeIndex)}_\PlayerIndex(\StratProfile) - \epsilon&
        &\textsf{By definition}\\
        &>(\gamma^* + \epsilon - \hat{\epsilon}^{(\TimeIndex)}_\PlayerIndex(\StratProfile)) + \hat{\epsilon}^{(\TimeIndex)}_\PlayerIndex(\StratProfile) - \epsilon&
        &\textsf{Pruning criterion}\\
        &= \gamma^*.
    \end{align*}
    Since $(\PlayerIndex, \StratProfile)$ was arbitrary, we have (w.h.p.) that all regret-pruned indices have positive corresponding regret in the true game and greater than $\gamma^*$ corresponding regret in the empirical game. But since all indices have been pruned, this implies that, with probability at least $1 - \delta$, all indices $(\PlayerIndex, \StratProfile)\in\UtilityIndices$ with $\Regret_\PlayerIndex(\StratProfile, \GameTuple_{\Simulator}) = 0$ or $\Regret_\PlayerIndex(\StratProfile, \EmpiricalGame{\Simulator}) \leq \gamma^*$ will be well-estimated pruned by \PSRegNU{} prior to termination, and will hence satisfy $\abs{\Utility_\PlayerIndex(\StratProfile) - \EUtility_\PlayerIndex(\StratProfile)}\leq \epsilon$. The correctness result then follows from \Cref{lem:new-dual-containment}.
\end{proof}

\regretU*
\begin{proof}
    Since \PSRegU{} uses uniform utility deviation bounds, we have that for each index $(\PlayerIndex, \StratProfile)\in\UtilityIndices$ and iteration $\TimeIndex$, it holds that $\hat{\epsilon}^{(\TimeIndex)}_\PlayerIndex(\StratProfile) = \hat{\epsilon}^{(\TimeIndex)}$. If we applied the regret pruning technique from \PSRegNU{} to such an algorithm, the pruning criterion would simplify as follows:
    \begin{align*}
        & & \Regret^\downarrow_\PlayerIndex(\StratProfile; \EUtility^{(t)}) &> \max\{0, \gamma^* + \epsilon - \hat{\epsilon}^{(t)}_\PlayerIndex(\StratProfile)\}\\
        &\Longleftrightarrow & \sup_{\StratProfile'\in\Adjacent_{\PlayerIndex, \StratProfile}}\left(\EUtility^{(\TimeIndex)}_\PlayerIndex(\StratProfile') - \hat{\epsilon}^{(\TimeIndex)}_\PlayerIndex(\StratProfile')\right) - \left(\EUtility^{(\TimeIndex)}_\PlayerIndex(\StratProfile) + \hat{\epsilon}^{(\TimeIndex)}_\PlayerIndex(\StratProfile)\right) &> \max\{0, \gamma^* + \epsilon - \hat{\epsilon}^{(t)}_\PlayerIndex(\StratProfile)\}\\
        &\Longleftrightarrow & \left[\sup_{\StratProfile'\in\Adjacent_{\PlayerIndex, \StratProfile}}\EUtility^{(\TimeIndex)}_\PlayerIndex(\StratProfile') - \EUtility^{(\TimeIndex)}_\PlayerIndex(\StratProfile)\right] - 2\hat{\epsilon}^{(\TimeIndex)} &> \max\{0, \gamma^* + \epsilon - \hat{\epsilon}^{(t)}\}\\
        &\Longleftrightarrow & \sup_{\StratProfile'\in\Adjacent_{\PlayerIndex, \StratProfile}}\EUtility^{(\TimeIndex)}_\PlayerIndex(\StratProfile') - \EUtility^{(\TimeIndex)}_\PlayerIndex(\StratProfile)  &> 2\hat{\epsilon}^{(\TimeIndex)} + \max\{0, \gamma^* + \epsilon - \hat{\epsilon}^{(t)}\}\\
        &\Longleftrightarrow & 
        \Regret_\PlayerIndex(\StratProfile; \EUtility^{(\TimeIndex)})&> \max\{2\hat{\epsilon}^{(\TimeIndex)}, \gamma^* + \epsilon + \hat{\epsilon}^{(t)}\}\enspace.
    \end{align*}
    But this is precisely the pruning criterion of \PSRegU{}. Thus, \PSRegNU{} is simply a generalization of \PSRegU{} to cases with non-uniform utility deviation bounds, and \PSRegU{} must then yield the same guarantees as \PSRegNU{}.
\end{proof}

\regretOld*
\begin{proof}
    Consider the regret pruning criterion in \PSRegU{} when $\gamma^* = 0$. We know that the first iteration $\TimeIndex$ on which $\hat{\epsilon}^{(\TimeIndex)}\leq \epsilon$, all indices will be well-estimated pruned and the algorithm will terminate. Thus, regret pruning will only occur on iterations on which $\hat{\epsilon}^{(\TimeIndex)} > \epsilon$. But then the regret pruning criterion simplifies to $\Regret_\PlayerIndex(\StratProfile; \EUtility^{(\TimeIndex)}) > \max\{2\hat{\epsilon}^{(\TimeIndex)}, \epsilon + \hat{\epsilon}\} = 2\hat{\epsilon}^{(\TimeIndex)}$, which is precisely the pruning condition of \PSRegOld{}. Thus, \PSRegU{} is simply a generalization of \PSRegOld{} to $\gamma^* > 0$, and hence \PSRegOld{} must yield the same guarantees as \PSRegU{} when $\gamma^* = 0$.
\end{proof}

\regretM*
\begin{proof}
    Suppose an index $(\PlayerIndex, \StratProfile)\in\UtilityIndices$ is regret-pruned on iteration $i$. In the proof of \Cref{thm:regret-pruning-2}, we see that $\Regret^\downarrow_\PlayerIndex(\StratProfile, \EmpiricalGame{\Simulator}^{(i)}) > 0$ implies that $\Regret_\PlayerIndex(\StratProfile, \GameTuple_{\Simulator})$. Hence, this pruning criteria is also guaranteed not to regret prune any index $(\PlayerIndex, \StratProfile')\in\UtilityIndices$ satisfying $\Regret_\PlayerIndex(\StratProfile, \GameTuple_{\Simulator}) = 0$. Further following the proof of \Cref{thm:regret-pruning-2}, we have that $\Regret_\PlayerIndex(\StratProfile, \EmpiricalGame{\Simulator}) \geq \Regret^\downarrow_\PlayerIndex(\StratProfile, \EmpiricalGame{\Simulator}^{(i)}) + \hat{\epsilon}^{(i)}_\PlayerIndex(\StratProfile) - \epsilon$, which when combined with our pruning criteria, yields (w.h.p.) that
    \begin{equation*}
        \Regret_\PlayerIndex(\StratProfile, \EmpiricalGame{\Simulator}) \geq 2\hat{\epsilon}^{(i)}_\PlayerIndex(\StratProfile) \geq 2\abs{\Utility_\PlayerIndex(\StratProfile) - \EUtility_\PlayerIndex(\StratProfile)}.
    \end{equation*}
    Since $(\PlayerIndex, \StratProfile)$ was arbitrary, we have that $\abs{\Utility_\PlayerIndex(\StratProfile) - \EUtility_\PlayerIndex(\StratProfile)}\leq \frac{\Regret_\PlayerIndex(\StratProfile, \EmpiricalGame{\Simulator})}{2}$ for each index $(\PlayerIndex, \StratProfile)\in\UtilityIndices$ that is regret-pruned. Since the remaining indices must all be well-estimated, we have (w.h.p.) that for all $(\PlayerIndex, \StratProfile)\in\UtilityIndices$, it holds that
    \begin{equation*}
        \abs{\Utility_p(\StratProfile) - \EUtility_p(\StratProfile)} \leq \max\left\{\epsilon, \frac{\Regret_\PlayerIndex(\StratProfile, \EmpiricalGame{\Simulator})}{2}\right\}.
    \end{equation*}
    The conclusion then follows from \Cref{lem:middle-dual-containment}.
\end{proof}

\subsection{Efficiency Bounds}

\citep{cousins2022computational} derive high-probability sample complexity bounds for their empirical Bennett tail bounds. We state these sample complexity results below, and use them to derive our efficiency bounds for \PSRegNU{} and \PSRegM{}.

\begin{lemma}
    \label{lem:sample-complexity}
    Consider an index $(\PlayerIndex, \StratProfile)\in\UtilityIndices$. If the sample size $\NumberOfSamples_{\StratProfile} \geq 2 + 2\ln\frac{3}{\delta} \left(\frac{5\UtilityRange}{2\epsilon} + \frac{\UtilityVariance_\PlayerIndex(\StratProfile)}{\epsilon^2}\right)$,
    then with probability at least $1 - \frac{\delta}{3}$, it will hold that $\epsilonEB_\PlayerIndex(\StratProfile) \leq \epsilon$.
\end{lemma}

\efficiency*
\begin{proof}
    Each of our efficiency bounds is presented as a minimum over two bounds, the first corresponding to regret pruning and the second to well-estimated pruning. It is clear that the second is a direct consequence of \Cref{lem:sample-complexity}. We show the regret pruning bounds, beginning with \PSRegNU{}.

    Recall that \PSRegNU{} prunes an index $(\PlayerIndex, \StratProfile)\in\UtilityIndices$ on an iteration $\TimeIndex$ if $\Regret^\downarrow_\PlayerIndex(\StratProfile; \EUtility^{(\TimeIndex)}) > \max\{0, \gamma^* + \epsilon - \hat{\epsilon}^{(\TimeIndex)}_\PlayerIndex(\StratProfile)\}$. We have (w.h.p.) that
    \begin{align*}
        \Regret^\downarrow_\PlayerIndex(\StratProfile; \EUtility^{(\TimeIndex)}) &= \sup_{\StratProfile'\in\Adjacent_{\PlayerIndex, \StratProfile}}\left(\EUtility^{(\TimeIndex)}_\PlayerIndex(\StratProfile') - \hat{\epsilon}^{(\TimeIndex)}_\PlayerIndex(\StratProfile')\right) - \left(\EUtility^{(\TimeIndex)}_\PlayerIndex(\StratProfile) + \hat{\epsilon}^{(\TimeIndex)}_\PlayerIndex(\StratProfile)\right)\\
        &> \left[\sup_{\StratProfile'\in\Adjacent_{\PlayerIndex, \StratProfile}}\EUtility^{(\TimeIndex)}_\PlayerIndex(\StratProfile')\right] - \left[\sup_{\StratProfile'\in\Adjacent_{\PlayerIndex, \StratProfile}}\hat{\epsilon}_\PlayerIndex^{(\TimeIndex)}(\StratProfile')\right] - \left(\EUtility^{(\TimeIndex)}_\PlayerIndex(\StratProfile) + \hat{\epsilon}^{(\TimeIndex)}_\PlayerIndex(\StratProfile)\right)\\
        &> \left[\sup_{\StratProfile'\in\Adjacent_{\PlayerIndex, \StratProfile}}\left(\Utility_\PlayerIndex(\StratProfile) - \hat{\epsilon}^{(\TimeIndex)}_\PlayerIndex(\StratProfile)\right)\right] - \left[\sup_{\StratProfile'\in\Adjacent_{\PlayerIndex, \StratProfile}}\hat{\epsilon}_\PlayerIndex^{(\TimeIndex)}(\StratProfile')\right] - \left(\Utility_\PlayerIndex(\StratProfile) + 2\hat{\epsilon}^{(\TimeIndex)}_\PlayerIndex(\StratProfile)\right)\\
        &> \left[\sup_{\StratProfile'\in\Adjacent_{\PlayerIndex, \StratProfile}}\Utility_\PlayerIndex(\StratProfile')\right] - 2\left[\sup_{\StratProfile'\in\Adjacent_{\PlayerIndex, \StratProfile}}\hat{\epsilon}_\PlayerIndex^{(\TimeIndex)}(\StratProfile')\right] - \left(\Utility_\PlayerIndex(\StratProfile) + 2\hat{\epsilon}^{(\TimeIndex)}_\PlayerIndex(\StratProfile)\right)\\
        &= \Regret_\PlayerIndex(\StratProfile; \Utility) - 2\hat{\epsilon}^{(\TimeIndex)}_\PlayerIndex(\StratProfile) - 2\sup_{\StratProfile'\in\Adjacent_{\PlayerIndex, \StratProfile}}\hat{\epsilon}^{(\TimeIndex)}_\PlayerIndex(\StratProfile')\enspace.
    \end{align*}
    Hence, a strictly tighter pruning criterion for \PSRegNU{} would be pruning an index $(\PlayerIndex, \StratProfile)\in\UtilityIndices$ when
    \begin{align*}
        \Regret_\PlayerIndex(\StratProfile; \Utility) - \gamma^* &> 2\hat{\epsilon}^{(\TimeIndex)}_\PlayerIndex(\StratProfile) + 2\sup_{\StratProfile'\in\Adjacent_{\PlayerIndex, \StratProfile}}\hat{\epsilon}^{(\TimeIndex)}_\PlayerIndex(\StratProfile') - \gamma^* + \max\{0, \gamma^* + \epsilon - \hat{\epsilon}^{(\TimeIndex)}_\PlayerIndex(\StratProfile)\}\\
        &= 2\sup_{\StratProfile'\in\Adjacent_{\PlayerIndex, \StratProfile}}\hat{\epsilon}^{(\TimeIndex)}_\PlayerIndex(\StratProfile') + \max\left\{2\hat{\epsilon}^{(\TimeIndex)}_\PlayerIndex(\StratProfile) - \gamma^*, \epsilon + \hat{\epsilon}^{(\TimeIndex)}_\PlayerIndex(\StratProfile)\right\}\enspace.
    \end{align*}
    We can make the pruning criterion even tighter by increasing the right-hand side:
    \begin{align*}
        2\sup_{\StratProfile'\in\Adjacent_{\PlayerIndex, \StratProfile}}\hat{\epsilon}^{(\TimeIndex)}_\PlayerIndex(\StratProfile') + \max\left\{2\hat{\epsilon}^{(\TimeIndex)}_\PlayerIndex(\StratProfile) - \gamma^*, \epsilon + \hat{\epsilon}^{(\TimeIndex)}_\PlayerIndex(\StratProfile)\right\} &\leq 4\sup_{\StratProfile'\in\Adjacent_{\PlayerIndex, \StratProfile}}\hat{\epsilon}^{(\TimeIndex)}_\PlayerIndex(\StratProfile').
    \end{align*}
    Hence, the latest (w.h.p.) an index $(\PlayerIndex, \StratProfile)\in\UtilityIndices$ will be regret-pruned by \PSRegNU{} is when
    \begin{equation*}
        \sup_{\StratProfile'\in\Adjacent_{\PlayerIndex, \StratProfile}}\hat{\epsilon}^{(\TimeIndex)}_\PlayerIndex(\StratProfile') < \frac{\Regret_\PlayerIndex(\StratProfile; \Utility) - \gamma^*}{4}.
    \end{equation*}
    Our result then follows via \Cref{lem:sample-complexity}. By analogous reasoning, we have that a strictly tighter regret pruning criterion than the one in \PSRegM{} would be pruning an index $(\PlayerIndex, \StratProfile)\in\UtilityIndices$ when
    \begin{align*}
        \Regret_\PlayerIndex(\StratProfile; \Utility) - \epsilon > 3\hat{\epsilon}^{(\TimeIndex)}_\PlayerIndex(\StratProfile) + 2\sup_{\StratProfile'\in\Adjacent_{\PlayerIndex, \StratProfile}}\hat{\epsilon}^{(\TimeIndex)}_\PlayerIndex(\StratProfile').
    \end{align*}
    Similar to before, we can make the pruning criterion even tighter by increasing the right hand side:
    $$3\hat{\epsilon}^{(\TimeIndex)}_\PlayerIndex(\StratProfile) + 2\sup_{\StratProfile'\in\Adjacent_{\PlayerIndex, \StratProfile}}\hat{\epsilon}^{(\TimeIndex)}_\PlayerIndex(\StratProfile') \leq 5\sup_{\StratProfile'\in\Adjacent_{\PlayerIndex, \StratProfile}}\hat{\epsilon}^{(\TimeIndex)}_\PlayerIndex(\StratProfile')$$
    Hence, the latest (w.h.p.) an index $(\PlayerIndex, \StratProfile)\in\UtilityIndices$ will be regret pruned by \PSRegM{} is when
    \begin{equation*}
        \sup_{\StratProfile'\in\Adjacent_{\PlayerIndex, \StratProfile}}\hat{\epsilon}^{(\TimeIndex)}_\PlayerIndex(\StratProfile') < \frac{\Regret_\PlayerIndex(\StratProfile; \Utility) - \epsilon}{5}\enspace.
    \end{equation*}
    Once again, our result follows from \Cref{lem:sample-complexity}.
\end{proof}

\omegaBound*
\begin{proof}
    Recall that all the aforementioned progressive sampling algorithms use well-estimated pruning. By Hoeffding's Inequality (\Cref{thm:hoeffding}), we have that on iteration $\ScheduleLength$ of algorithm PS, for all $(\PlayerIndex, \StratProfile)\in\UtilityIndices$, it holds that
    \begin{align*}
        \hat{\epsilon}^{(\ScheduleLength)}_\PlayerIndex(\StratProfile) \leq \UtilityRange\sqrt{\frac{\ln\left(\frac{2\abs{\UtilityIndices}\ScheduleLength}{\delta}\right)}{\CumulativeSamples_{\ScheduleLength}}} \leq \UtilityRange \sqrt{\ln\left(\frac{2\abs{\UtilityIndices}\ScheduleLength}{\delta}\right)\cdot \frac{2\epsilon^2}{\UtilityRange^2\ln\left(\frac{2\abs{\UtilityIndices}\ScheduleLength}{\delta}\right)}} = \epsilon,
    \end{align*}
    and thus all indices which remain active until iteration $\ScheduleLength$ will be pruned on iteration $\ScheduleLength$. Hence, the aforementioned algorithms are guaranteed to prune all indices prior to the exhaustion of the sampling schedule, and thus will return an empirical game satisfying the respective guarantees of the algorithm.
\end{proof}

\subsection{Sampling Schedule}

Our sampling schedule is derived via a sample complexity lower bound for the empirical Bennett tail bounds presented in \Cref{thm:eBennett}. \citep{cousins2022computational} lower bound the empirical Bennett bounds via the zero-variance case of Bennett's inequality (\Cref{thm:bennett}). We, however, use a tighter lower bound which we derive below.

\begin{lemma}
    \label{lem:sample-complexity-lower}
    Consider an index $(\PlayerIndex, \StratProfile)\in\UtilityIndices$. If $\epsilonEB_\PlayerIndex(\StratProfile) \leq \epsilon$, then it must hold that the sample size $\NumberOfSamples_{\StratProfile} > \left(\frac{1}{3} + \sqrt{\frac{4 + 2\sqrt{3}}{3}}\right)\cdot \frac{\UtilityRange\ln \left( \frac{3 \abs{\UtilityIndices}}{\delta} \right)}{\epsilon}$.
\end{lemma}
\begin{proof}
    We have that
    \begin{align*}
        \EVarianceEpsilon_{\PlayerIndex}(\StratProfile)
        &= \frac{2 \UtilityRange^{2} \ln \left( \frac{3 \abs{\UtilityIndices}}{\delta} \right)}{3 \NumberOfSamples} + \sqrt{\left(\frac{1}{3} + \frac{1}{2\ln\left(\frac{3\abs{\UtilityIndices}}{\delta}\right)}\right) \left( \frac{\UtilityRange^{2} \ln \left( \frac{3\abs{\UtilityIndices}}{\delta} \right)}{\NumberOfSamples - 1}\right)^{2} + \frac{2 \UtilityRange^{2} \EUtilityVariance_\PlayerIndex(\StratProfile) \ln \left( \frac{3\abs{\UtilityIndices}}{\delta} \right)}{\NumberOfSamples}}\\
        &> \frac{2 \UtilityRange^{2} \ln \left( \frac{3 \abs{\UtilityIndices}}{\delta} \right)}{3 \NumberOfSamples} + \sqrt{\frac{1}{3} \left( \frac{\UtilityRange^{2} \ln \left( \frac{3\abs{\UtilityIndices}}{\delta} \right)}{\NumberOfSamples - 1}\right)^{2}}\\
        &> \frac{2 + \sqrt{3}}{3}\cdot\frac{\UtilityRange^{2} \ln \left( \frac{3 \abs{\UtilityIndices}}{\delta} \right)}{\NumberOfSamples}\enspace.
    \end{align*}
    We further have that
    \begin{align*}
        \epsilonEB_\PlayerIndex(\StratProfile)
        &= \frac{\UtilityRange\ln \left( \frac{3 \abs{\UtilityIndices}}{\delta} \right)}{3 \NumberOfSamples} + \sqrt{\frac{2 \left(\EUtilityVariance_\PlayerIndex (\StratProfile) + \EVarianceEpsilon_{\PlayerIndex}( \StratProfile)\right) \ln \left( \frac{3\abs{\UtilityIndices}}{\delta} \right)}{\NumberOfSamples}}\\
        &> \frac{\UtilityRange\ln \left( \frac{3 \abs{\UtilityIndices}}{\delta} \right)}{3 \NumberOfSamples} + \sqrt{\frac{2\cdot \frac{2 + \sqrt{3}}{3}\cdot\frac{\UtilityRange^{2} \ln \left( \frac{3 \abs{\UtilityIndices}}{\delta} \right)}{\NumberOfSamples}\cdot \ln \left( \frac{3\abs{\UtilityIndices}}{\delta} \right)}{\NumberOfSamples}}\\
        &= \left(\frac{1}{3} + \sqrt{\frac{4 + 2\sqrt{3}}{3}}\right)\cdot \frac{\UtilityRange\ln \left( \frac{3 \abs{\UtilityIndices}}{\delta} \right)}{\NumberOfSamples}
    \end{align*}
    The conclusion follows directly.
\end{proof}

Thus, our sampling schedule for \PSWE{} begins at $\alpha = \left(\frac{1}{3} + \sqrt{\frac{4 + 2\sqrt{3}}{3}}\right)\cdot \frac{\UtilityRange\ln \left( \frac{3 \abs{\UtilityIndices}\ScheduleLength}{\delta} \right)}{\epsilon}$ and ends at a cumulative sample size that is at least $\omega\doteq \frac{\UtilityRange^2\ln \left( \frac{3 \abs{\UtilityIndices}\ScheduleLength}{\delta}\right)}{2\epsilon^2}$ to satisfy \Cref{thm:omega}. Following \citep{cousins2022computational}, we then use a schedule with a geometrically increasing cumulative sample size with a geometric factor $\beta$ (for our experiments, we use $\beta = 1.1$). Our schedule length is then $\ScheduleLength \doteq \lceil\log_\beta \left(\frac{\omega}{\alpha}\right)\rceil$. The first sample size is defined by $\NumberOfSamples_1 \doteq \alpha\beta$, and each following sample size is defined by $\NumberOfSamples_{\TimeIndex} \doteq \alpha\beta^{\TimeIndex} - \NumberOfSamples_{\TimeIndex - 1}$.

Of course, for all of our regret pruning algorithms, it may be possible for regret pruning to occur prior to at least one index $(\PlayerIndex, \StratProfile)\in\UtilityIndices$ achieving $\hat{\epsilon}^{(\TimeIndex)}_\PlayerIndex(\StratProfile) \leq \epsilon$. Regret pruning cannot, however, happen prior to at least one index $(\PlayerIndex, \StratProfile)\in\UtilityIndices$ achieving $\hat{\epsilon}^{(\TimeIndex)}_\PlayerIndex(\StratProfile) \leq \frac{\UtilityRange}{2}$. This can be seen by looking at the loosest regret pruning criterion we discuss, that used in \PSRegNU{} with $\gamma^* = 0$. An index $(\PlayerIndex, \StratProfile)\in\UtilityIndices$ is regret-pruned on iteration $\TimeIndex$ if $\Regret^\downarrow_\PlayerIndex(\StratProfile; \EUtility^{(\TimeIndex)}) > 0$. Notice that if $\hat{\epsilon}^{(\TimeIndex)}_\PlayerIndex(\StratProfile) > \frac{\UtilityRange}{2}$ for all $(\PlayerIndex, \StratProfile)\in\UtilityIndices$, then for any given index $(\PlayerIndex, \StratProfile)\in\UtilityIndices$, we have that
\begin{align*}
    \Regret^\downarrow_\PlayerIndex(\StratProfile; \EUtility^{(\TimeIndex)}) &= \sup_{\StratProfile'\in\Adjacent_{\PlayerIndex, \StratProfile}}\left(\EUtility^{(\TimeIndex)}_\PlayerIndex(\StratProfile') - \hat{\epsilon}^{(\TimeIndex)}_\PlayerIndex(\StratProfile')\right) - \left(\EUtility^{(\TimeIndex)}_\PlayerIndex(\StratProfile) + \hat{\epsilon}^{(\TimeIndex)}_\PlayerIndex(\StratProfile)\right)\\
    &< \Regret_\PlayerIndex(\StratProfile; \EUtility^{(\TimeIndex)}) - \UtilityRange,
\end{align*}
and hence $(\PlayerIndex, \StratProfile)$ will be regret-pruned only if it holds that $\Regret_\PlayerIndex(\StratProfile; \EUtility^{(\TimeIndex)}) - \UtilityRange > \Regret^\downarrow_\PlayerIndex(\StratProfile; \EUtility^{(\TimeIndex)}) > 0$. But the latter is impossible since $\Regret_\PlayerIndex(\StratProfile; \EUtility^{(\TimeIndex)}) \leq \UtilityRange$ by definition. Hence, no index can be regret-pruned (by any of our regret pruning criteria) prior to at least one index being achieving an estimation guarantee of at least $\frac{\UtilityRange}{2}$.

Based on the above, we start our sampling schedule for all our regret pruning algorithms on $\alpha' \doteq \left(\frac{1}{3} + \sqrt{\frac{4 + 2\sqrt{3}}{3}}\right)\cdot 2\ln \left( \frac{3 \abs{\UtilityIndices}\ScheduleLength}{\delta} \right)$. But, as argued in the text, using a schedule with a strictly geometrically increasing schedule, we waste too many iterations on small sample sizes and yield a schedule length that is too large. Hence, we instead fix the schedule length to be $1.5$ times the schedule length used for \PSWE{}. We then occupy the final two-thirds of our schedule with the same sample sizes used in the sampling schedule for \PSWE{}, and occupy the first third of our schedule with one that has linearly increasing cumulative sample size beginning at $\alpha'$ and ending at $\alpha$ (from above).

\end{document}